\documentclass[11pt,a4paper]{article}
\usepackage{times}

\usepackage[T1]{fontenc}
\usepackage{times}
\usepackage{soul}
\usepackage{url}
\usepackage[hidelinks]{hyperref}
\usepackage[utf8]{inputenc}
\usepackage[small]{caption}
\usepackage{graphicx}
\usepackage{amsmath}
\usepackage{amsthm}
\usepackage{booktabs}
\usepackage{algorithm}
\usepackage{algorithmic}
\usepackage{authblk}
\usepackage{fullpage}
\usepackage{caption}
\usepackage{subcaption}

\usepackage[numbers]{natbib}

\usepackage{mathtools}
\usepackage{kbordermatrix}
\usepackage{hyperref}
\usepackage[utf8]{inputenc}
\usepackage{amsmath,amsthm,amsfonts}
\usepackage{mathtools}
\usepackage[dvipsnames]{xcolor}
\usepackage{amssymb}
\usepackage{pifont}
\newcommand{\cmark}{\ding{51}}%
\newcommand{\xmark}{\ding{55}}%
\usepackage{caption}
\usepackage{subcaption}

\usepackage{thmtools} 
\usepackage{thm-restate}

\usepackage{booktabs} \usepackage{nicefrac} \usepackage{chngpage}
\usepackage{todonotes} \usepackage{cleveref} \usepackage{placeins}

\newtheorem{proposition}{Proposition}
\newtheorem{observation}{Observation}
\newtheorem{claim}{Claim}

\theoremstyle{definition}
\newtheorem{example}{Example}

\newtheorem{definition}{Definition}
\newtheorem{lemma}{Lemma}
\newtheorem{theorem}{Theorem}

\newcommand{\POS}{{{\mathrm{POS}}}}
\newcommand{\PAIR}{{{\mathrm{PAIR}}}}
\newcommand{\VL}{{{\mathrm{VL}}}}
\newcommand{\AVL}{{{\mathrm{AVL}}}}
\newcommand{\BOR}{{{\mathrm{Borda}}}}
\newcommand{\EMD}{{{\mathrm{emd}}}}
\newcommand{\emd}{{{\mathrm{emd}}}}

\newcommand{\ID}{{{\mathrm{ID}}}}
\newcommand{\UN}{{{\mathrm{UN}}}}
\newcommand{\AN}{{{\mathrm{AN}}}}
\newcommand{\ST}{{{\mathrm{ST}}}}

\newcommand{\calL}{\mathcal{L}}
\newcommand{\calF}{\mathcal{F}}
\newcommand{\calP}{\mathcal{P}}
\newcommand{\calS}{\mathcal{S}}
\newcommand{\calB}{\mathcal{B}}
\newcommand{\calM}{\mathcal{M}}

\newcommand{\np}{{{\mathrm{NP}}}}

\newcommand{\sort}{{{\mathrm{sort}}}}
\newcommand{\disc}{{{\mathrm{disc}}}}
\newcommand{\swap}{{{\mathrm{swap}}}}
\newcommand{\pair}{{{\mathrm{pair}}}}

\newcommand{\demdpos}{{{d_\mathrm{pos}^{\mathrm{emd}}}}}

\newcommand{\dellpos}{{{d_\mathrm{pos}^{\ell_1}}}}

\newcommand{\pref}{\succ}

\newcommand{\pos}{{\mathrm{pos}}}
\usepackage{amsthm}

\newtheorem{manualtheoreminner}{Theorem}
\newenvironment{manualtheorem}[1]{%
  \renewcommand\themanualtheoreminner{#1}%
  \manualtheoreminner
}{\endmanualtheoreminner}

\makeatletter
\def\moverlay{\mathpalette\mov@rlay}
\def\mov@rlay#1#2{\leavevmode\vtop{%
   \baselineskip\z@skip \lineskiplimit-\maxdimen
   \ialign{\hfil$\m@th#1##$\hfil\cr#2\crcr}}}
\newcommand{\charfusion}[3][\mathord]{
    #1{\ifx#1\mathop\vphantom{#2}\fi
        \mathpalette\mov@rlay{#2\cr#3}
      }
    \ifx#1\mathop\expandafter\displaylimits\fi}
\makeatother

\newcommand{\cupdot}{\charfusion[\mathbin]{\cup}{\cdot}}

\newcommand\numberthis{\addtocounter{equation}{1}\tag{\theequation}}

\newcommand{\pfnote}[1]{}
\newcommand{\nbnote}[1]{}
\newcommand{\ssnote}[1]{}
\newcommand{\twnote}[1]{}
\usepackage{etoolbox}

\newcommand{\lv}[1]{}

\newcommand{\appendixText}{}

\allowdisplaybreaks
\title{Understanding Distance Measures Among Elections\thanks{An 
extended abstract of this article has been accepted for publication in the 
proceedings of IJCAI 2022.}}

\author[1]{Niclas Boehmer}
\author[2]{Piotr Faliszewski}
\author[1]{Rolf Niedermeier}
\author[2]{Stanis\l{}aw Szufa}
\author[2,3]{Tomasz W\k{a}s}

\affil[1]{\small
  Technische Universit\"at Berlin, Algorithmics and Computational 
  Complexity\protect\\
  niclas.boehmer@tu-berlin.de}
 \affil[2]{\small
  AGH University,
  faliszew@agh.edu.pl,szufa@agh.edu.pl}
 \affil[3]{\small
  University of Warsaw,
  t.was@mimuw.edu.pl}
\date{\today}

\begin{document}

\maketitle

\begin{abstract}
  Motivated by putting empirical work based on (synthetic) election data
  on a more solid mathematical basis,
  we analyze six distances among elections, including, e.g., the
  challenging-to-compute but very precise swap distance and the
  distance used to form the so-called \emph{map of elections}.
  Among the six, the latter 
  seems to strike the best balance between its computational
  complexity and expressiveness. 

\end{abstract}

\section{Introduction}
We study the properties of several
distances (metrics) among elections.\footnote{Generally, we use the
  word \emph{distance} when we refer to a value of a \emph{metric},
  but occasionally, reflecting the literature, we break this rule (e.g., to speak of the earth
  mover's distance or the swap distance).}
We focus on the ordinal model, 
where each voter ranks the candidates from the most to the least
appealing one, and on metrics that are independent of renaming the
candidates and voters.
Such 
metrics were introduced by \citet{fal-sko-sli-szu-tal:c:isomorphism},
who argued about their usefulness to compare two elections
generated from some statistical models, or to evaluate which
statistical model is most likely to generate elections similar to
given real-life ones. Indeed, in such applications the names of the
candidates and voters do not carry any information and should be
disregarded. 
Unfortunately, while the metrics studied by
\citet{fal-sko-sli-szu-tal:c:isomorphism} are naturally
motivated---for example, one of them extends the widely accepted swap
distance---many of them are not only $\np$-hard to compute (even
approximately), but also difficult to compute in practice.

Yet, many of the motivating ideas of 
\citet{fal-sko-sli-szu-tal:c:isomorphism} were soon implemented in the
\emph{maps of elections}, introduced by
\citet{szu-fal-sko-sli-tal:c:map} and extended by
\citet{boe-bre-fal-nie-szu:c:compass}.
Briefly put, such a map is a collection of
election instances, typically generated from some statistical models
(but \citet{boe-bre-fal-nie-szu:c:compass} also used real-life ones
from PrefLib~\cite{mat-wal:c:preflib}), together with their
distances. The elections are represented graphically as points in the
plane, whose Euclidean distances 
are as similar to the distances between the respective elections as
possible.
Since maps of elections typically contain hundreds of elections,
instead of using the appealing-but-hard-to-compute extension of the
swap distance, \citet{szu-fal-sko-sli-tal:c:map} introduced and used a much
simpler metric.
Yet, the maps proved to be quite useful. 
For example, \citet{szu-fal-sko-sli-tal:c:map}
used their map to find 
hard instances for some multiwinner voting
rule, 
\citet{boe-bre-fal-nie-szu:c:compass} and \citet{DBLP:journals/corr/abs-2204-03589} obtained insights about
different types of real-life elections, and
\citet{boe-bre-fal-nie:c:counting-bribery} used the maps to study the
robustness of elections.

In this work, we take a step back and analyze the properties of
several metrics, including those of
\citet{fal-sko-sli-szu-tal:c:isomorphism} and
\citet{szu-fal-sko-sli-tal:c:map}.
Our goal is to help putting empirical work with election data on a
more solid mathematical basis and, in particular, to understand if
basing the election maps 
on the simpler metric was a good decision, what was its price, and if
one should have used other metrics.
We view the swap
distance as a yardstick against which we measure the other ones.
In particular, we consider the
following issues:
\begin{enumerate}  
\item As a basic test, we compare the metrics' ability to distinguish
  nonisomorphic elections.  Since some metrics act on 
  aggregate representations of elections, they may sometimes fail at
  this task. We also study the complexity of computing an election
  with a given representation.

\item We analyze distances between four ``compass'' elections, which
  capture four types of (dis)agreement among the
  voters~\cite{boe-bre-fal-nie-szu:c:compass}. We find that two of
  them are the most distant elections under each of our metrics.

\item We compute the correlation between the values provided by the
  swap distance and the other metrics; we also compare the maps that
  they produce.

\item We note that the swap distance can be understood in terms of
  the shortest paths on a certain graph and we analyze to what extent this
  applies to the other metrics.

\end{enumerate}
Some proofs and arguments are in the appendix.

\section{Preliminaries}

For an integer $n$, let $[n] := \{1, \ldots, n\}$.

\paragraph{Preference Orders.}
Let $C$ be 
set of candidates. By~$\calL(C)$ we denote the set of all total
orders over $C$, referred to either as preference orders or votes,
depending on the context.

\paragraph{Elections.}
An election $E = (C,V)$ consists of a candidate
set~$C = \{c_1, \ldots, c_m\}$ and a preference
profile~$V = (v_1, \ldots, v_n)$, where each $v_i$ is a preference
order from $\calL(C)$.
For example, a preference order $v_i \colon a \pref b \pref c$
indicates that the $i$-th voter ranks candidate $a$ highest, followed
by candidates~$b$ and~$c$. Given a vote $v$ and a candidate $c$, we
write $\pos_v(c)$ to denote the position of~$c$ in~$v$ (the top-ranked
candidate has position $1$, the next one has position $2$, and so on).

\begin{example}\label{ex:1}
  Consider elections $E = (C,V)$ and $E' = (C',V')$, where
  $C = \{a,b,c\}$, $C' = \{x,y,z\}$, $V = (v_1,v_2,v_3)$,
  $V' = (v'_1,v'_2,v'_3)$, and the votes are:
  \begin{align*}
    \small
    v_1 \colon a \pref b \pref c, &&
    v_2 \colon b \pref c \pref a, && 
    v_3 \colon b \pref a \pref c, \\
    v'_1 \colon x \pref y \pref z, &&
    v'_2 \colon x \pref y \pref z, && 
    v'_3 \colon y \pref x \pref z.                                                         \end{align*}
\end{example}

\paragraph{Aggregate Representations.}
Let $E = (C,V)$ be an election.
We use the following 
aggregate representations of~$E$:
\begin{enumerate}
\item For each two candidates $c, d \in C$, $\calM_E(c,d)$ is the
  number of voters that prefer~$c$ to $d$ in election~$E$. We call it
  the weighted majority relation and represent it as an $m \times m$
  matrix where rows and columns correspond to the candidates (the
  diagonal is undefined).
  A relative weighted majority relation is a weighted majority
  relation from whose entries we subtract $\nicefrac{n}{2}$ (let $n$
  be even here).
    
\item For a candidate $c \in C$ and a position $i \in [m]$, 
  $\calP_E(c,i)$ is the number of voters from $E$ that rank $c$ on
  position~$i$; $\calP_E(c) = (\calP_E(c,1), \ldots, \calP_E(c,m))$ is
  a (column) position vector of $c$.
  We view $\calP_E$ as a matrix with columns
  $\calP_E(c_1), \ldots, \calP_E(c_m)$ and call it a position matrix.

\item For a candidate $c \in C$,
  $\calB_E(c) := \sum_{i=1}^n \big(m-\pos_{v_i}(c)\big)$ is the Borda
  score of $c$ in $E$. Then, $\calB_E$ is the Borda score vector,
  whose entries correspond to the candidates.
  
\end{enumerate}

\begin{example}\label{ex:2}
  Below we provide $\calM_E$, $\calP_E$, and $\calB_E$, respectively,
  for election~$E$ from Example~\ref{ex:1}:
  \begin{align*}
    \small
    \kbordermatrix{ & a & b & c   \\
    a\!\!\!\! &               - & 1 & 2\\
    b\!\!\!\! &               2 & - & 3\\
    c\!\!\!\! &               1 & 0 & -\\
    },&&
         \small
         \kbordermatrix{ & a & b & c   \\
    1\!\!\!\! &               1 & 2 & 0\\
    2\!\!\!\! &               1 & 1 & 1\\
    3\!\!\!\! &               1 & 0 & 2\\
    },&&
         \small
         \kbordermatrix{ & a & b & c   \\
                    &               3 & 5 & 1\\
    }.
  \end{align*}
\end{example}

\paragraph{(Pseudo)Metrics.}
Given a set $X$, a function
$d \colon X \times X \rightarrow \mathbb{R}$ is a pseudometric over
$X$ if for each three elements $a, b, c \in X$ it holds that
(i)~$d(a,b) = d(b,a) \geq 0$, (ii)~$d(a,a) = 0$, and
(iii)~$d(a,c) \leq d(a,b) + d(b,c)$.  For brevity, we will refer to
our pseudometrics as metrics (formally, a metric should assume
value~$0$ only if both its arguments are identical).

\paragraph{Metrics Among Vectors.}
Let $x = (x_1, \ldots, x_n)$ and $y = (y_1, \ldots, y_n)$ be
two real-valued vectors. Then, 
$ \ell_1(x,y) := |x_1-y_1| + \cdots + |x_n-y_n|$
is the \emph{$\ell_1$-metric} between $x$ and~$y$.
Given a real-valued vector $z = (z_1, \ldots, z_n)$, we
write~$\hat{z}$ to denote its prefix-sum variant, i.e., an
$n$-dimensional vector such that for each $i \in [n]$, its $i$-th
entry is $\hat{z}_i = z_1 + z_2 + \cdots +  z_i$.
If the entries of $x$ and~$y$ sum up to the same value and contain only nonnegative
entries, then their \emph{earth mover's distance} is defined as~\cite{rubner2000earth}:
\[
  \EMD(x,y) := \ell_1(\hat{x},\hat{y}).
\]
Alternatively, $\emd(x,y)$
is defined as the minimum total cost of a sequence of operations that
transform vector~$x$ into vector~$y$, where each operation is of the
form ``subtract value~$\alpha$ from some $x_i$ (where, at the time of
performing the operation, we have $x_i \geq \alpha$) and add $\alpha$
to $x_j$'' and has cost $\alpha \cdot |j-i|$.
Both definitions are equivalent~\cite{rubner2000earth}.

\paragraph{Bijections.}  Given two equal-sized sets $X$ and $Y$, let
$\Pi(X,Y)$ denote the set of all one-to-one mappings from $X$ to $Y$.
For a positive integer $n$, let $S_n$ be the set of all permutations
of $[n]$ (i.e., $S_n = \Pi([n],[n])$).
Given two equal-sized candidate sets $C$ and $D$, a preference order
$v \in \calL(C)$, and a function $\sigma \in \Pi(C,D)$, we
write~$\sigma(v)$ to denote the preference order obtained from $v$ by
replacing each candidate $c \in C$ with the candidate $\sigma(c) \in
D$. Given a preference profile
$V = (v_1, \ldots, v_n) \in (\calL(C))^n$, by $\sigma(V)$ we mean
$(\sigma(v_1), \ldots, \sigma(v_n))$.  We use analogous notation for
other objects defined over candidate sets.  For example, for an
election $E = (C,V)$, we write $\sigma(\calM_E)$ to denote the
weighted majority relation of the election $(\sigma(C),\sigma(V)) = (D,\sigma(V))$.

\section{Metrics Among Elections}\label{sec:distances}

In this section, we define the six metrics among elections that we
study. All but the last one already appeared in the literature.  We
only consider distances between elections with the same numbers of
candidates and the same numbers of voters.

\subsection{Swap and Discrete Isomorphic Metrics} 
Let $C$ be a candidate set and let $u,v \in \calL(C)$ be two votes.
Their discrete distance, $d_\disc(u,v)$, is $0$ if they are identical
and it is~$1$ otherwise. Their swap distance, $d_\swap(u,v)$, is the
number of inversions between $u$ and $v$, i.e., the number of pairs of
candidates $c, d \in C$ such that $u$ and $v$ rank these candidates in
opposite order.

Let $d$ be either $d_\disc$ or $d_\swap$. For two elections,
$E = (C,V)$ and $E' = (C',V')$, where $|C| = |C'|$,
$V = (v_1, \ldots, v_n)$, and $V' = (v'_1, \ldots, v'_n)$,
\citet{fal-sko-sli-szu-tal:c:isomorphism} extended $d$ to elections as
follows:
\[ 
  \textstyle
  d(E,E') := \min_{\sigma \in \Pi(C,C')} \min_{\rho \in S_n} \sum_{i=1}^n d\big( \sigma(v_i), v'_{\rho(i)}\big).
\]
In other words, under the extended distance we match the candidates
and the votes of the two input elections so that the sum of the
distances between the matched votes is minimal.
\begin{example}\label{ex:4}
  The discrete distance between elections $E$ and~$E'$ from
  Example~\ref{ex:1} is one, as witnessed by the candidate matching
  $\sigma(a) = x$, $\sigma(b) = y$, and $\sigma(c) = z$.
  For the same matching,
  the swap distance between $E$ and $E'$ is two.
\end{example}

We refer to the extensions of $d_\disc$ and $d_\swap$ as the discrete
and swap \emph{isomorphic metrics}. This stems from the fact that two
elections are isomorphic (i.e., can be made identical by renaming
candidates and reordering the votes) if and only if their discrete
distance (equivalently, swap distance) is zero.

\citet{fal-sko-sli-szu-tal:c:isomorphism} have shown that while
computing the discrete isomorphic distance can be done in polynomial
time, the same task for the swap distance is $\np$-hard (in essence,
this follows from the hardness proofs for the Kemeny voting
rule~\cite{bar-tov-tri:j:who-won,dwo-kum-nao-siv:c:rank-aggregation}).

\subsection{Positionwise and Pairwise Metrics}
To circumvent the hardness of computing the isomorphic swap distance,
\citet{szu-fal-sko-sli-tal:c:map} introduced two other metrics, based
on analyzing aggregate representations of elections.

Let $E = (C,V)$ and $E' = (C',V')$ be two elections such that
$|C| = |C'|$ and $|V|= |V'|$. The EMD-positionwise distance between $E$
and $E'$ is:
\begin{equation*}\label{eq:demdpos}
  \demdpos(E,E') \!:= \min_{\sigma \in \Pi(C,C')} \sum_{c \in C}\EMD\big(
  \calP_E(c), \calP_{E'}(\sigma(c)\big).
\end{equation*}
Note that we view each candidate as her or his position vector
and we seek a candidate matching that minimizes the earth mover's distances
between the vectors of matched~candidates.
\begin{example}\label{ex:5} 
  Consider the same elections and the same matching $\sigma$ as in
  Example~\ref{ex:4}. We have that $\calP_E(a) = (1,1,1)$, whereas
  $\calP_{E'}(\sigma(a)) = (2,1,0)$ and their earth mover's distance
  is $2$. Altogether, $\demdpos(E,E')=2+1+1=4$.
\end{example}
\citet{szu-fal-sko-sli-tal:c:map} based their metric on EMD because
they felt it was intuitively appropriate.  We aim to verify this
intuition and, thus, we also consider the $\ell_1$-positionwise
metric, which uses the $\ell_1$ distance instead of EMD.

\citet{szu-fal-sko-sli-tal:c:map} also introduced the pairwise metric,
which works on top of the weighted majority relation:
\[
  d_\pair(E,E') := \!\!\!\!\min_{\sigma \in \Pi(C,C')} \!\!
  \sum_{\substack{a, b \in C \\ a \neq b}} \!\! |\calM_E(a,b) -
  \calM_{E'}(\sigma(a),\sigma(b))|. 
\]
\begin{example} 
  Take the elections from Example~\ref{ex:1} and matching
  $\sigma'(a) = y$, $\sigma'(b) = x$, and $\sigma'(c) = z$; $\calM_E$
  and $\sigma'(\calM_{E'})$ are:
  \begin{align*}
    \small
    \kbordermatrix{ & a & b & c   \\
    a\!\!\!\! &               - & 1 & 2\\
    b\!\!\!\! &               2 & - & 3\\
    c\!\!\!\! &               1 & 0 & -\\
    },&&
         \small
         \kbordermatrix{ & \sigma'(a) & \sigma'(b) & \sigma'(c)   \\
    \sigma'(a)\!\!\!\! &               - & 1 & 3\\
    \sigma'(b)\!\!\!\! &               2 & - & 3\\
    \sigma'(c)\!\!\!\! &               0 & 0 & -\\
    }
  \end{align*}
  and we see that the pairwise distance of our elections (for this
  matching) is $0+1+0+0+1+0 = 2$.
\end{example}

The positionwise distances can be computed in polynomial
time~\cite{szu-fal-sko-sli-tal:c:map}, whereas the pairwise distance
is $\np$-hard to compute
\cite{gro-rat-woe:c:approximate-isomorphism,szu-fal-sko-sli-tal:c:map}).

\subsection{Bordawise Metric}
We introduce one new metric, similar in spirit to the positionwise and
pairwise ones, but defined on top of the election's Borda score vectors.  Given
two equal-sized elections $E$ and $E'$, their Bordawise
distance is defined as:
\[
  d_\BOR(E,E') := \emd( \sort(\mathcal{B}_E), \sort(\calB_{E'}) ),
\]
where for a vector $x$, $\sort(x)$ means a vector obtained from~$x$ by
sorting it in nonincreasing order.
The Bordawise metric is defined to be as simple as possible, while trying
to still be meaningful. For example, 
sorting the score vectors ensures that two isomorphic elections are at
distance zero and removes the use of an explicit matching between the
candidates.

\begin{example}
  The distance between the elections from Example~\ref{ex:1} is
  $\emd\big((5,3,1),(5,4,0) \big) = 1$.
\end{example}

\section{Aggregate Representations}\label{sec:aggregate}

First, we discuss how the aggregate representations that underlie our
metrics affect their ability to distinguish nonisomorphic elections.
Then, we study the complexity of deciding if a given representation
indeed corresponds to some election.

\begin{table}[t]
    \centering
    \begin{tabular}{ c | c | c | c | c }
      $|C|\times|V|$ & ANECs & Positionwise  & Pairwise & Bordawise\\
    	\midrule
        $3 \times 3$    & 10 & 10 & 8 & 8 \\
        $3 \times 4$    & 24 & 23 & 17 & 13 \\
        $3 \times 5$    & 42 & 40 & 25 & 18 \\
    	\midrule
        $4 \times 3$    & 111 & 93 & 50 & 37 \\
        $4 \times 4$    & 762 & 465 & 200 & 76  \\
        $4 \times 5$    & 4095 & 1746 & 513 & 131  \\
       
        \end{tabular}
        \caption{Number of equivalence classes under our metrics.}
    \label{table:num_classes}
\end{table}

Given a metric, two elections are in the same equivalence class if
their distance is zero.
An \emph{anonymous, neutral equivalence class (ANEC)} consists of all
isomorphic elections with a given number of candidates and
voters~\cite{ege-gir:j:isomorphism-ianc}.
While ANECs are 
the equivalence classes of the swap and
discrete metrics, the other metrics are less precise and
their equivalence classes are unions of some ANECs. 

To get a feeling as to how much precision is lost
due to various aggregate representations,
in Table~\ref{table:num_classes} we compare the number of ANECs and
the numbers of equivalence classes 
of the positionwise, pairwise, and Bordawise metrics, for small
elections; we computed the table using exhaustive
search\footnote{There are exact formulas for some columns in
  Table~\ref{table:num_classes}, but not for all. See, e.g., the work
  of \citet{ege-gir:j:isomorphism-ianc}.}  (note that EMD- and
$\ell_1$-positionwise metrics have the same equivalence classes).
Among these metrics, positionwise ones perform best and Bordawise performs
worst.
Next, we provide a partial theoretical
explanation for this observation.

We say that a metric~$d$ is at least as fine as a metric~$d'$ if for
each two elections $A$ and~$B$, $d(A,B) = 0$ implies that
$d'(A,B) = 0$ (i.e., each equivalence class of $d$ is a subset of
some equivalence class of~$d'$). Metric~$d$ is finer
than $d'$ if it is at least as fine as $d'$ but $d'$ is not at least
as fine as $d$.

\begin{restatable}{proposition}{diamond}
\label{pr:diamond}
Swap and discrete isomorphic metrics are finer than
  EMD/$\ell_1$-positionwise and pairwise, which both are finer than
  Bordawise. Neither EMD/$\ell_1$-positionwise is finer than pairwise
  nor the other way round.
\end{restatable}
\begin{proof}
  Since two isomorphic elections are at distance zero according to
  each of our metrics, we see that swap and discrete isomorphic
  metrics are at least as fine as all the other ones.  Next, let $d$
  be one of the EMD/$\ell_1$-positionwise metrics or the pairwise
  metric. For each two elections $A$ and $B$ (with the same number of
  candidates and the same number of voters), if $d(A,B) = 0$ then
  either the position matrices of $A$ and $B$ are isomorphic (if $d$
  is a positionwise metric) or their weighted majority relations are
  (if $d$ is the pairwise metric). Since the Borda score vectors can
  be computed both from the position matrices and from the weighted
  majority relations, it must be that $d_\BOR(A,B) = 0$.

  Let us now show that EMD/$\ell_1$-positionwise are not as fine as
  the swap and discrete ones. We form election~$X$ with preference
  profile:
  \begin{align*}
    \small
    x_1 \colon a \pref b \pref c, &&
    x_2 \colon b \pref a \pref c, &&
    x_3 \colon c \pref a \pref b, \\                                                         x_4 \colon a \pref c \pref b, &&
    x_5 \colon b \pref c \pref a, &&
    x_6 \colon c \pref b \pref a.
  \end{align*}
  That is, $X$ contains each possible vote from $\calL(\{a,b,c\})$.
  We also from election $Y$ with the following preference profile:
  \begin{align*}
    \small
    y_1 \colon a \pref b \pref c, &&
    y_2 \colon b \pref c \pref a, &&
    y_3 \colon c \pref a \pref b, \\                                                         y_4 \colon a \pref b \pref c, &&
    y_5 \colon b \pref c \pref a, &&
    y_6 \colon c \pref a \pref b.
  \end{align*}
  The two elections are not isomorphic (e.g., because all the votes in
  $X$ are distinct, but this is not the case for $Y$), but under
  positionwise metrics their distance is zero. This is so because
  their position matrices are identical:
  \[
    \calP_X = \calP_Y =
             \small
         \kbordermatrix{ & a & b & c   \\
    1\!\!\!\! &               2 & 2 & 2\\
    2\!\!\!\! &               2 & 2 & 2\\
    3\!\!\!\! &               2 & 2 & 2\\
    }.
  \]
  Since we also have that the pairwise distance between $X$ and $Y$ is
  nonzero, it holds that the positionwise metrics are not at least as
  fine as the pairwise one. To see that the pairwise distance between $X$
  and $Y$ is nonzero, note that:
  \begin{align*}
    \calM_X = \kbordermatrix{ & a & b & c   \\
    a\!\!\!\! &               - & 3 & 3\\
    b\!\!\!\! &               3 & - & 3\\
    c\!\!\!\! &               3 & 3 & -\\
    }, &&
    \calM_Y = \kbordermatrix{ & a & b & c   \\
    a\!\!\!\! &               - & 4 & 2\\
    b\!\!\!\! &               2 & - & 4\\
    c\!\!\!\! &               4 & 2 & -\\
    }.
  \end{align*}
  Since all entries of $\calM_X$ are equal but this is not the case
  for $\calM_Y$, the two elections must be at nonzero pairwise
  distance.

  To see that the pairwise metric is not as fine as the swap and
  discrete isomorphic ones, let us consider election $Z$ with
  preference profile:
  \begin{align*}
    \small
    z_1 \colon a \pref b \pref c, &&
    z_2 \colon a \pref b \pref c, &&
    z_3 \colon a \pref b \pref c, \\                                                         z_4 \colon c \pref b \pref a, &&
    z_5 \colon c \pref b \pref a, &&
    z_6 \colon c \pref b \pref a.
  \end{align*}
  Clearly, $X$ and $Z$ are not isomorphic. Further, their
  EMD/$\ell_1$-positionwise distances are nonzero, but their pairwise
  distance is zero. To see why this is the case, note that:
  \begin{align*}
    \calM_X = \calM_Z =  \kbordermatrix{ & a & b & c   \\
    a\!\!\!\! &               - & 3 & 3\\
    b\!\!\!\! &               3 & - & 3\\
    c\!\!\!\! &               3 & 3 & -\\
    }, &&
    \calP_Z = 
             \small
         \kbordermatrix{ & a & b & c   \\
    1\!\!\!\! &               3 & 0 & 3\\
    2\!\!\!\! &               0 & 6 & 0\\
    3\!\!\!\! &               3 & 0 & 3\\
    }.
  \end{align*}
  Since $\calP_X$ and $\calP_Z$ contain different values (e.g., the
  latter contains value $6$ and the former does not) it must be that
  the positionwise distances between $X$ and $Z$ are nonzero.
  Pairwise distance between $X$ and $Z$ is zero because their weighted
  majority relations are identical.
\end{proof}

Nonetheless, having many equivalence classes does not automatically
make a metric desirable.  For example, the discrete isomorphic metric
is as fine as the swap one,
but has many unappealing properties.

Let us now move to the problem of deciding if a given matrix/vector
indeed represents some election.
For position matrices,
 \citet{boe-bre-fal-nie-szu:c:compass} have shown this to be easy:
 If a matrix $M$ has nonnegative entries and its
rows and columns sum up to the same value, then there is a
polynomial-time computable election $E$ with
$M = \calP_E$. 
Hence, one can work directly on the matrices and recover elections when
needed.  Unfortunately, for Borda score vectors and weighted majority
relations analogous problems are $\np$-complete. Thus, by operating on
them directly, we may leave the space of elections.

\begin{restatable}{theorem}{recBorda}
\label{thm:recBorda}
Given a vector $x$ of nonnegative integers, it is $\np$-complete to
decide if there is an election~$E$ with~$\calB_E = x$.
\end{restatable}
\begin{proof}
  \citet{DBLP:journals/scheduling/YuHL04} showed that given a sequence
  of positive integers $a_1, \dots, a_m$ such that
  $a_1 \geq a_2 \geq \cdots \geq a_m$, $\sum_{i=1}^m a_i=m(m+1)$, and
  for each $i \in [m]$ we have $2\leq a_i \leq 2m$, it is
  $\np$-complete to decide if there are two permutations $\phi, \phi' \in S_m$ such that for all $i \in [m]$ it holds that
  $\phi(i)+\phi'(i)=a_i$. We reduce this problem to the one from the
  statement of the theorem by forming a vector
  $x = (a_1-2, \ldots, a_m-2)$.

  If there are two permutations $\phi$ and $\phi'$ that satisfy the
  conditions of \citeauthor{DBLP:journals/scheduling/YuHL04}'s
  problem, then we form a two-voter election $E = (C,V)$ as follows:
  We let $C = [m]$ and we form two votes, $v$ and $v'$.  For each
  candidate $i \in C$, the first (the second) voter ranks $i$ on
  position $m-\phi(i)+1$ ($m-\phi'(i)+1$); note that the produced
  votes rank exactly one candidate in each position because $\phi$ and
  $\phi'$ are permutations. Then, the Borda score of each $i \in C$ is
  $\phi(i)-1+\phi'(i)-1=a_i-2$.

  For the other direction, assume that there is an election
  $E = (C,V)$ with Borda score vector $x$. Then, $E$ must contain
  exactly two voters because otherwise the sum of the candidates'
  scores would either be too large or too small. W.l.o.g., we assume
  that $C = [m]$ and that each candidate $i \in C$ has Borda score
  $a_i-2$. Let $v$ and $v'$ be the two votes in $E$.  We form a
  permutation $\phi$ so that for each $i \in C$ we have
  $\phi(i) = m-\pos_v(i)+1$, We form $\phi'$ analogously, but using
  $v'$ instead of $v$. It follows that for each $i \in [m]$ we have
  $\phi(i)+\phi'(i)=(a_i-2)+2=a_i$. This completes the proof.
\end{proof}
On the positive side, there is an FPT-algorithm for the above problem,
parameterized by the number of candidates (note that we take the
number $n$ as part of the input for simplicity; it can always be
concluded from the sum of the required Borda scores). 
\begin{proposition} \label{pr:recogFPT}
  There is an algorithm that given a vector $y = (y_1, \ldots, y_m)$
  of nonnegative integers and a number~$n$, decides if there is an
  election~$E$ with~$m$ candidates and~$n$ voters in time
  $\mathcal{O}(m^{2m+1}\cdot n)$.
\end{proposition}

\begin{proof}
  We construct an Integer Linear Program (ILP) to solve the problem.
  Let $C = [m]$ be a set of $m$ candidates.  For each
  $v\in \mathcal{L}(C)$, we introduce a variable $x_v$ that denotes
  the number of copies of $v$ in the election to be constructed.  For
  each candidate $i\in C$, $y_i$ is the Borda score that~$i$ should
  end up with.  Our ILP consists of the following constraints, one for
  each candidate $i\in C$:
  \[
    \textstyle \sum_{v\in \mathcal{L}(C)} x_v\cdot (m-\pos_v(i))=y_i.
  \]
  It is immediate that each solution to the ILP corresponds to an
  election realizing the given score vector and the other way around.
  Concerning the running time, 
  \citet{DBLP:journals/talg/EisenbrandW20} proved that an ILP with
  $t$ constraints, where each entry in the constraint matrix is upper
  bounded by $\Delta$ and each entry of the vector on the right hand
  side is upper bounded by $\Gamma$ can be solved in time
  $\mathcal{O}((t\cdot \Delta)^t\cdot\Gamma^2)$.  Thus, our problem is
  solvable in $\mathcal{O}(m^{2m+1}\cdot n)$.
\end{proof}

The hardness of deciding whether there is an election with a given Borda
score vector is not too surprising because it is closely related to
strategic voting under the Borda rule, which is also
$\np$-complete~\cite{bet-nie-woe:c:borda,dav-kat-nar-wal-xia:j:strategic-voting-borda}.
The case of weighted majority relation is more intriguing because the
classic McGarvey's theorem~\cite{mcg:j:election-graph} gives a
polynomial-time algorithm for recovering an election with a given
\emph{relative} weighted majority relation (but see also the
work of \citet{bac-bra-gei-har-kar-pet-see:j:weighted-majority-relation}).

\begin{restatable}{theorem}{recPair}
\label{thm:recPair}
  Given an $m \times m$ matrix $M$, it is $\np$-complete to decide if
  there is an election $E$ with $\mathcal{M}_E=M$.
\end{restatable}
\begin{proof}
  We reduce from the NP-complete \textsc{Restricted X3C} problem
  (\textsc{RX3C}) \cite{gon:j:x3c}. Its instances consist of a
  universe $X=\{x_1,\dots, x_{3t}\}$ and a family
  $\mathcal{S}=\{S_1,\dots , S_{3t}\}$ of size-$3$ subsets of~$X$,
  where each~$x_i$ appears in exactly three sets from~$\mathcal{S}$.
  We ask if there is a family $\mathcal{S}'\subseteq \mathcal{S}$ of
  $t$ sets such that $\bigcup_{S_i \in \calS'}S_i = X$ (i.e., we ask
  if there is an exact cover of~$X$).
  Let $(X,\calS)$ be an instance of RX3C. We form
  a candidate set:
  \[
    C = \{d\} \cup X \cup \{s_1. \ldots, s_{3t}\} \cup \{b_1,b_2,b_3\}
    \cup A \cup F,
  \]
  where $A = \{a_1, \ldots, a_{3t}\}$ and
  $F = \{f_1, \ldots, f_{3t}\}$ are sets of ``location''
  candidates. The candidates $s_1, \ldots, s_{3t}$ correspond to the sets
  from $\calS$, $b_1, b_2, b_3$ will delineate blocks in the votes,
  and $d$ will be used to encode the exact cover (if one exists). For
  each $i \in [3t]$, by $A(i)$ and $F(i)$ we mean the orderings:
  \begin{align*}
    A(i) \colon & a_1\succ \dots \succ a_{i} \succ s_i \succ a_{i+1}
    \succ \dots \succ a_{3t},\\
    F(i) \colon & f_1\succ s_1 \succ f_2 \succ s_2 \succ \dots \succ f_{i-1} \succ s_{i-1}\\ & \quad\quad\quad\; \succ f_{i}\succ f_{i+1}\succ s_{i+1} \succ \dots f_{3t}\succ s_{3t}.
  \end{align*}  
  For a subset $C'\subseteq C$ of candidates, $[C']$~denotes an
  arbitrary, fixed ordering of the candidates from $C'$.  For every
  $S_{i}=\{x_j,x_k,x_{\ell}\}\in \mathcal{S}$, let $v_{i}$ denote the
  following vote:
  \begin{align*}
    & b_1\succ x_j\succ x_k\succ x_{\ell}\succ b_2 \succ  A(i) \succ [X\setminus S_i] \succ b_3 \succ F(i).
  \end{align*}
  Let $E' = (C \setminus \{d\},(v_1, \ldots, v_{3t}))$ be an election.
  We form a weighted majority relation $M$ by taking $\calM_{E'}$ and
  extending it to include $d$ as follows:
  We require that $d$ is placed behind $b_1$ in exactly $t$ votes, $d$
  is placed behind each candidate from $X$ in exactly one vote, and
  $d$ is never placed behind any other candidate.
  We claim that there is an election $E$ such that $M = \calM_E$ if
  and only if $\calS$ contains an exact cover of~$X$.  \medskip
  
  \noindent $\boldsymbol{(\Rightarrow)}$ First, let $\mathcal{S}'\subseteq \mathcal{S}$ be an exact cover of
  $X$. To construct $E$, for each $S_i\in \mathcal{S'}$, we include
  vote $v_{i}$ with $d$ inserted right in front of~$b_2$, and for each
  $S_i\in \mathcal{S}\setminus \mathcal{S'}$, we include $v_{i}$ with
  $d$ inserted before $b_1$. As $\mathcal{S}'$ is an exact cover, 
  $M = \calM_E$.
  \medskip

  \noindent $\boldsymbol{(\Leftarrow)}$ It remains to prove that if there is an election realizing $M$, then there is a solution to the given \textsc{X3C} instance.
  Let $E = (C,V)$ be an election such that
  $M = \calM_E$.  From now on, speaking of votes, we mean the
  votes appearing in~$E$.  We claim the following.
  
  \begin{claim} \label{claim1} In each vote from~$E$, exactly one
    $s_i$ ranks before~$b_3$.
  \end{claim}
  \begin{proof}
    We make the following observations, which follow directly from the
    definition of~$M$:
  \begin{enumerate}
  \item We have
    $b_2\succ a_1\succ a_2\succ \dots \succ a_{3t}\succ b_3\succ
    f_1\succ f_2 \succ \dots f_{3t}$ in each vote. \label{it:1}
  \item Each $s_i$ is placed before $b_3$ in exactly one vote, behind
    $f_i$ in all but one vote, and before $f_{i+1}$ in each vote.  Further,
    $s_i$ is placed behind $a_i$ in all votes and before $a_{i+1}$ in
    one.
  \end{enumerate}
  From these observations we conclude that in one vote $s_i$ is placed
  before~$b_3$ and between $a_i$ and $a_{i+1}$, and in all other
  votes $s_i$ is placed behind~$b_3$, between $f_i$ and~$f_{i+1}$.

  Now, for the sake of contradiction, assume that there is a vote
  where both $s_i$ and $s_j$, with $i<j$, are placed before~$b_3$. By
  the above discussion and by \Cref{it:1} from the above list, in this
  vote it holds that $a_i\succ s_i \succ \dots \succ a_j\succ s_j$ and
  in all other votes it holds that
  $f_i\succ s_i \succ \dots \succ f_j\succ s_j$.  Thus, $s_i$ is ranked
  ahead of~$s_j$ in all the votes.  But $M$ requires that there is one
  vote where $s_j$ is ranked ahead of~$s_i$ (due to vote~$v_j$ in the
  construction of~$M$). This proves the claim.
  \end{proof}

  \noindent Next, we make the following observations based on $M$:
  \begin{enumerate}
  \item Each $s_i$ is placed behind $b_2$ in all votes and before
    $b_3$ in exactly one vote. \label{it:1a}
  \item Each $x_\ell \in X$ is placed before $b_3$ in all
    votes. \label{it:1b}
  \item For each $x_\ell\in X$ and each $S_i\in \mathcal{S}$
    not containing~$x_\ell$, $x_\ell$ is ranked behind $s_i$ in exactly
    one vote. \label{it:2}
  \item Candidate $d$ is placed in the first position in $2t$ votes
    ($b_1$ is always placed before all candidates from
    $C\setminus \{d,b_1\}$ and $d$ is placed before $b_1$ in $2t$
    votes), and before $b_2$ in all votes. \label{it:3}
  \item For each $x_\ell\in X$, there is exactly one vote where
    $x_\ell$ is placed before $d$. \label{it:4}
  \end{enumerate}
  By \Cref{it:1a,it:1b,it:2} from the above list and by \Cref{claim1},
  we conclude that for each $S_i\in \mathcal{S}$, in the one vote
  where $s_i$ is placed before $b_3$ (but still behind $b_2$), $s_i$
  is also placed before all candidates from $X\setminus S_i$,
  implying that candidates from $X \setminus S_i$ are placed behind
  $b_2$ in this vote.

  Let $V'\subseteq V$ be the subset of those votes where $d$ is not in
  the first position.  Let $\mathcal{S}'$ be the family containing
  those sets $S_i$ for which candidate $s_i$ is placed before $b_3$ in
  a vote from $V'$.  We claim that $\mathcal{S}'$ is an exact cover.
  First, by \Cref{claim1} and by \Cref{it:3} from the above list, it
  follows that $\mathcal{S}'$ contains exactly $t$ sets.  By
  \Cref{claim1}, it follows that in each vote from~$V'$ there is a
  candidate $s_i$ corresponding to a member of $\mathcal{S}'$ ranked
  before $b_3$.  For the sake of contradiction, assume that there is
  an $x_\ell\in X$ such that $x_\ell$ is not part of any set from
  $\mathcal{S'}$.  Then, by the preceding paragraph, it holds that
  $x_\ell$ is placed behind~$b_2$ in all votes from~$V'$.  As by
  \Cref{it:3} from the above list $d$ is always placed before~$b_2$,
  $d$ is placed before~$x_\ell$ in all the votes from~$V'$.  As $d$ is
  placed in the first position in all other votes, $d$ is always
  ahead of $x_\ell$, contradicting the definition of~$M$.
  \end{proof}

\section{Diameter and Compass Elections} \label{sec:compass}
In this section, we analyze the distances between four ``compass''
elections of \citet{boe-bre-fal-nie-szu:c:compass}. These elections
capture four different types of (dis)agreement among the voters and,
thus, we expect good metrics to put them far apart.
Since the smallest nonzero distances under all our metrics are either
$1$, $2$, or $4$, the larger are the distances between the compass
elections, the more space there is between them for other elections.
For technical reasons, we fix the number~$m$ of candidates to be even,
and the number of voters to be $n = t\cdot m!$, where~$t$ is some
positive integer (we will see ways to relax this assumption). The
compass elections are defined as follows:

\begin{enumerate}
\item In the identity elections, denoted $\ID$, all voters have the
  same, fixed preference order. 
\item In the antagonism elections, denoted $\AN$, half of the voters
  rank the candidates in one way and half of the voters rank them in
  the opposite way.

\item In the uniformity elections, denoted $\UN$, each possible vote
  appears the same number of times.
\item In the stratification elections, denoted $\ST$, the candidates
  are partitioned into two equal-sized sets $A$ and~$B$.  Each
  possible preference order where all members of~$A$ are ranked ahead
  of~$B$ appears the same number of times.
\end{enumerate}

\begin{table*}[t]
\begin{adjustwidth}{-2in}{-2in}
    \centering
    \scriptsize
    \begin{tabular}{l | c | c | c  | c | c | c }
                      & Discrete & Swap & $\ell_1$-Pair. & $\ell_1$-Pos.     & EMD-Pos. & EMD-Borda. \\
        \midrule
        $d(\ID,\UN), \ \ m!|n$    & $n \frac{m!-1}{m!}$   & $\frac{1}{4}n(m^2-m)$     & $\frac{1}{2}n(m^2-m)$ & $2 n (m-1)$         & $\frac{1}{3} n (m^2-1)$   & $\frac{1}{12} n (m^3 -m)$ \\
        $d(\ID,\AN), \ \ 2|n$     & $\frac{1}{2} n$       & $\frac{1}{4}n(m^2-m)$     & $\frac{1}{2}n(m^2-m)$ & $\frac{1}{2} n m$   & $\frac{1}{4} n m^2$       & $\frac{1}{12} n (m^3 -m)$ \\
        $d(\ID,\ST), \ \ ((\frac{m}{2})!)^2|n$ & $n\frac{((m/2)!)^2-1}{((m/2)!)^2}$& $\frac{1}{8}n(m^2-2m)$ & $\frac{1}{4}n(m^2-2m)$& $2 n (m-2)$         & $\frac{1}{6} n (m^2-4)$   & $\frac{1}{48} n (m^3 \!\! + \! 3m^2 \!\! - \!4m)$\\
        $d(\UN,\AN), \ \ m!|n$    & $n \frac{m!-2}{m!}$     & $\Theta(nm^2)^\dagger$ & $0$               & $2 n (m-2)$   & $\frac{1}{6} n (m^2-4)$   & $0$\\
        $d(\UN,\ST), \ \ m!|n$    & $n \frac{m! - (m/2)!^2}{m!}$  & $\frac{1}{8}nm^2$ & $\frac{1}{4}nm^2$ & $\frac{1}{2}m$& $\frac{1}{4} n m^2$   & $\frac{1}{16} n (m^3 - m)$\\
        $d(\AN,\ST), \ \ ((\frac{m}{2})!)^2|n$ & $n\frac{((m/2)!)^2-1}{((m/2)!)^2}$         & $\Theta(nm^2)^\dagger$ & $\frac{1}{4}nm^2$ & $2 n (m-2)$   & $\frac{13}{48} n (m^2 -\frac{16}{13})$ & $\frac{1}{48} n (m^3 \!\! + \! 3m^2 \!\! - \!4m)$
        \end{tabular}
        
\end{adjustwidth}
\caption{Overview of the distances of the compass matrices for our metrics. In the leftmost column, we denote conditions on $m$ and $n$ that need to hold in order for the two respective compass elections to be well-defined and our formula to hold. $\dagger$ For $d_\swap(\UN,\AN)$ and $d_\swap(\AN,\ST)$, we do not have a closed form formula (and we are not sure if it exists), however it holds that $\nicefrac{1}{8} \ n(m^2 - 3m + 2) \le d_\swap(\UN,\AN) \le \nicefrac{1}{4} \ n(m^2 -m)$ and also $\nicefrac{1}{8} \ n(m^2 - 2m) \le d_\swap(\AN,\ST) \le \nicefrac{1}{4} \ n(m^2 -m)$. Results for EMD-positionwise come from \protect\citet{boe-bre-fal-nie-szu:c:compass}.}\label{tab:distances}
\end{table*}

The next proposition gives asymptotic distances between the compass
elections (we provide exact values in \Cref{tab:distances}).

\begin{restatable}{proposition}{compass}
\label{pr:compass}
 Let $X$ and $Y$ be two distinct compass elections. Then,
  $d_\BOR(X,Y) = \Theta(nm^3)$,
  $d_\swap(X,Y) = \demdpos(X,Y) = d_\pair(X,Y) = \Theta(nm^2)$,
  $\dellpos(X,Y) = \Theta(nm)$, and
  $d_\disc(X,Y) = \Theta(n)$,
  except that $d_\pair(\AN,\UN) = d_\BOR(\AN,\UN) = 0$.
\end{restatable}
The distances between the compass elections are the largest
under the swap and EMD-positionwise metrics, followed by those under
 $\ell_1$-positionwise and discrete isomorphic metrics.
Pairwise and Bordawise metrics perform particularly badly because they
cannot distinguish between $\UN$ and $\AN$.
Yet, except for this, the compass elections are asymptotically as
distant under each of our metrics as possible. Indeed, $\ID$ and $\UN$
even form diameters of our election spaces (this also confirms a
conjecture of \citet{boe-bre-fal-nie-szu:c:compass}). 
We present only the proof for EMD-positionwise here and defer the proofs for all other metrics to \Cref{app:diameter}. 

\begin{restatable}{theorem}{diameter}
\label{thm:diameter}
  Let $d$ be one of our six metrics.
  For each two elections $X$ and $Y$ (with sizes as specified at the
  beginning of this section) it holds that $d(X,Y) \leq d(\ID,\UN)$.
\end{restatable}
\begin{proof}[Proof (EMD-positionwise)]
  The EMD-positionwise metric can be seen as working over matrices whose
  entries are nonnegative and whose rows and columns sum up to the
  same value.
  Let $A$ and~$B$ be two such matrices, whose rows and columns sum up
  to some value $n'$. Let $A/n'$ and $B/n'$ be the same matrices, but
  with their entries divided by $n'$.  Note that
  $\demdpos(A,B) = n' \cdot \demdpos(A/n',B/n')$.  From now on, we
  focus on matrices whose entries are nonnegative and whose rows and
  columns sum up to $1$ (they are called \emph{bistochastic}).

  \citet{boe-bre-fal-nie-szu:c:compass} have shown that
  $\demdpos(\ID,\UN) = \nicefrac{n(m^2-1)}{3}$.  We claim that for
  each two $m \times m$ bistochastic matrices
  $X$ and $Y$ it holds that $\demdpos(X,Y) \leq (m^2 -1)/3$. For the
  sake of contradiction, assume that the opposite holds.  Let
  $x_1, \ldots, x_m$ be the columns of $X$ and $y_1, \ldots, y_m$ be
  the columns of $Y$. Without loss of generality, we assume that
  $\demdpos(X,Y) = \sum_{i \in [m]} \emd(x_i,y_i)$; otherwise we could
  reorder the columns of one of the matrices.  By definition of the
  EMD metric, we have that
  $\sum_{i \in [m]} \emd(x_i,y_i) \textstyle = \sum_{i \in [m]}
  \ell_1(\hat{x}_i,\hat{y}_i)$.  For each $i \in [m]$, we write
  $\hat{x}_{i,1}, \ldots, \hat{x}_{i,m}$ to denote the entries of the
  cumulative vector $\hat{x}_i$; we use analogous notation for
  $\hat{y}_i$.
  Note that in the matrices with columns
  $\hat{x}_1, \ldots, \hat{x}_m$ and $\hat{y}_1, \ldots, \hat{y}_m$,
  for each $j \in [m]$, the $j$-th row sums up to $j$ (we refer to
  this as the \emph{cumulative rows} property).  Using these
  observations, we note that:
  \begin{align*}
    \textstyle
    \sum_{i \in [m]} &\emd(x_i,y_i)  \textstyle
    = \sum_{i \in [m]} \ell_1(\hat{x}_i,\hat{y}_i)\\ 
    &= \textstyle \sum_{i,j \in [m]} |\hat{x}_{i,j} - \hat{y}_{i,j}| \notag\\
    &= \textstyle \sum_{i,j \in [m]} \big(\max(\hat{x}_{i,j},\hat{y}_{i,j}) - \min(\hat{x}_{i,j},\hat{y}_{i,j})\big) \notag\\
    &= \textstyle \sum_{i,j \in [m]} (\hat{x}_{i,j} + \hat{y}_{i,j}) - 2 \sum_{i,j \in [m]} \min (\hat{x}_{i,j},\hat{y}_{i,j}) \notag\\
    &= \textstyle 2 \left( \sum_{j \in [m]} j \right) - 2 \sum_{i,j \in [m]} \min (\hat{x}_{i,j},\hat{y}_{i,j}) \\
    &= m(m+1) - 2 \textstyle \sum_{i,j \in [m]} \min (\hat{x}_{i,j},\hat{y}_{i,j}).
  \end{align*}
  Thus, if $\demdpos(X,Y) > (m^2-1)/3$, then it must hold that:
  \begin{align*}
    \textstyle \sum_{i,j \in [m]} &\min (\hat{x}_{i,j},\hat{y}_{i,j})  < 
                                                   \textstyle \frac{1}{2} \left( m(m+1) - \frac{m^2-1}{3}\right)\\
                                                  & = (2m^2 + 3m +1)/6 = (2m+1)(m+1)/6.
  \end{align*}
  
  In the following, for each $i,k \in [m]$ with $i+k > m$, if $i+k$ is used as a column
  index, then we take it to be $i+k-m$ (i.e., column indices
  ``cycle'').
  For each $k \in [m]$, we have
  $\demdpos(X,Y) \leq \sum_{i \in [m]} \emd(x_i,y_{i+k})$; if this
  were not the case, then our assumption that
  $\demdpos(X,Y) = \sum_{i \in [m]} \emd(x_i,y_i)$ would have been
  false.
  Consequently, for every $k\in [m]$, repeating the above reasoning, we get: 
  \[
    \textstyle
    \sum_{i,j \in [m]} \min(\hat{x}_{i,j}, \hat{y}_{i+k,j}) < (2m+1)(m+1)/6,
  \]
  If $a, b \in [0,1]$ then $a \cdot b \le \min(a,b)$.  As for each
  $i, j \in [m]$ we have $\hat{x}_{i,j},\hat{y}_{i,j} \in [0,1]$, for
  each $k \in [m]$ we have:
  \[
    \textstyle \sum_{i,j \in [m]} \hat{x}_{i,j} \cdot \hat{y}_{i+k,j} <
    (2m+1)(m +1)/6.
  \]
  By summing this inequality sidewise for all $k \in [m]$, we get:
  \[
    \textstyle
    \sum_{i,j \in [m]} \hat{x}_{i,j} \cdot \sum_{k \in [m]} \hat{y}_{k,j} < m(2m+1)(m+1)/6.
  \]
  By applying the cumulative rows property, we obtain:
  \[
    \textstyle \sum_{j=1}^m j^2 < m(2m+1)(m+1)/6.
  \]
  Since we know that $\sum_{j=1}^m j^2 = {m(2m+1)(m+1)}/{6}$, this is
  a contradiction. Hence, for all $m \times m$ bistochastic matrices
  $X$, $Y$ we have $\demdpos(X,Y) \leq (m^2 -1)/3$. Thus, for all
  elections with $m$ candidates and $n$ voters, their EMD-positionwise
  distance is at most $n(m^2-1)/3$.
\end{proof}

In the above proof we do not work directly with elections, but,
rather, with normalized position matrices. Viewed this way, $\ID$ is a
unit diagonal matrix and $\UN$ is a matrix whose entries are all equal. Indeed, this is how \citet{boe-bre-fal-nie-szu:c:compass} defined
them. 
In this way, the proof
works for any number of voters
(this also applies to $\ell_1$-positionwise and, using normalized
weighted majority relations, to pairwise).

For Bordawise, the proof of Theorem~\ref{thm:diameter} shows that for
each election the sum of its distances from $\ID$ and $\UN$ is the
same. That is, under this metric every election lays on the diameter.

\begin{table}[t]
    \centering
    \begin{tabular}{ c | c | c | c | c | c }
        $|C| \times |V|$ & EMD-Pos.& $\ell_1$-Pos. & $\ell_1$-Pair. & Bordawise & Discrete\\
    	\midrule
        $3 \times 3$    & 0.942 & 0.748 & 0.860 & 0.587 & 0.614 \\
        $3 \times 4$    & 0.900 & 0.697 & 0.860 & 0.659 & 0.636 \\
        $3 \times 5$    & 0.920 & 0.759 & 0.843 & 0.606 & 0.680 \\ 
    	\midrule
        $4 \times 3$    & 0.850 & 0.577 & 0.735 & 0.442 & 0.402 \\
        $4 \times 4$    & 0.782 & 0.561 & 0.689 & 0.415 & 0.434 \\
        $4 \times 5$    & 0.772 & 0.567 & 0.672 & 0.439 & 0.432 \\[1mm]
    	\midrule
      $\substack{\mathrm{10 \times 50} \\ \text{(340 elections)}}$  \rule{0mm}{2.75mm}  & 0.745 & 0.563 & 0.708 & 0.430 & 0.342 \\
        \end{tabular}
        \caption{Pearson correlation coefficients between swap distances and the
        other ones computed for our datasets.}
    \label{table:pearson_correlation}
\end{table}

\section{Maps and Correlations} \label{sec:correlation}
While in the previous section we studied distances between
hand-crafted elections, now we analyze automatically-generated ones.
We test how our metrics correlate with the swap one, and we compare
their maps of elections.

We use two datasets. The first one consists of all small elections, as
in Section~\ref{sec:aggregate}. The second one resembles those used in
the maps of \citet{szu-fal-sko-sli-tal:c:map} and
\citet{boe-bre-fal-nie-szu:c:compass}, but consists of elections with
$10$~candidates and $50$~voters,\footnote{This is the largest size for which
  we could compute swap distances within a few
  weeks. \citet{szu-fal-sko-sli-tal:c:map} used $100 \times 100$
  elections, and \citet{boe-bre-fal-nie-szu:c:compass} used
  $10 \times 100$ ones.}
generated according to the following statistical models (see the
just-cited papers for more details):

\begin{description}
\item[IC, Urn, and Mallows] We generated 20 elections using
the impartial culture model (IC), where each vote is selected
uniformly at random, and 60 elections for each of the classic urn and
Mallows models (we used the same sampling protocol as
\citet{boe-bre-fal-nie-szu:c:compass}).
\item[SP, SC, and SPOC] We generated 20 single-peaked elections (SP elections) uniformly at random (this is
known as the SP Walsh model~\cite{wal:t:generate-sp}), 20 such elections using the
Conitzer model~\cite{con:j:eliciting-singlepeaked}, and 20~single-peaked on a circle
elections (SPOC elections), uniformly at random~\cite{pet-lac:j:spoc}. We also generated 20~single-crossing elections (SC elections) using
the sampling protocol of \citet{szu-fal-sko-sli-tal:c:map}.
\item[Euclidean]
In these elections, each candidate and voter is a point from
some Euclidean space
and the voters rank the candidates with respect to their increasing
distances from them. We have generated the points uniformly at random
from (i)~a 1D interval, (ii)~a 2D sphere, (iii)~a 2D disc, and (iv)~a
3D cube; in each case we generated $20$ elections.
\item[Group-Separable] Group-separable elections were introduced by
\citet{ina:j:group-separable,ina:j:simple-majority}. We use
a definition based on trees (see, e.g., the works
of~\citet{kar:j:group-separable}
and~\citet{fal-kar-obr:c:group-separable} for a discussion and
motivation). Consider an ordered, rooted tree where each leaf is a
unique candidate. To obtain a vote, for each of the nodes we can
choose to reverse the order of its children and, then, rank the
candidates by reading them off from left to right. Given such a tree,
we sample a vote by reversing the order of each node's children
with probability $\nicefrac{1}{2}$. We generated 20 elections using
complete binary trees, and 20 elections using binary caterpillar trees
(a binary caterpillar tree is defined recursively to either be a leaf,
or a root whose one child is a leaf and whose other child is a root of
a binary caterpillar tree.)
\end{description}

\begin{table}[t]
    \centering
    \small
    \begin{tabular}{ c | c | c | c | c | c }
        name & EMD-Pos.& $\ell_1$-Pos. & $\ell_1$-Pair. & EMD-Borda. & Discrete\\
    	\midrule
        Impartial Culture  & 0.481 & 0.114 & 0.525 & -0.064 & -0.039 \\
        SP by Conitzer  & 0.471 & 0.727 & -0.142 & 0.16 & 0.976 \\
        SP by Walsh  & 0.377 & 0.467 & -0.119 & -0.073 & 0.7 \\
        SPOC  & 0.297 & 0.409 & -0.074 & -0.065 & 0.622 \\
        Single-Crossing  & 0.252 & 0.248 & 0.123 & -0.007 & 0.625 \\
        1D Interval  & 0.242 & 0.219 & 0.101 & -0.08 & 0.606 \\
        2D Disc  & 0.337 & 0.317 & 0.203 & -0.002 & 0.636 \\
        3D Cube  & 0.406 & 0.347 & 0.311 & 0.035 & 0.67 \\
        2D Sphere  & 0.406 & 0.329 & 0.335 & 0.039 & 0.651 \\
        Urn (gamma)  & 0.84 & 0.86 & 0.803 & 0.713 & 0.102 \\
        Norm-Mallows (uniform)  & 0.86 & 0.784 & 0.839 & 0.793 & 0.255 \\
        GS Balanced  & 0.863 & 0.793 & 0.844 & 0.797 & 0.259 \\
        GS Caterpillar  & 0.864 & 0.795 & 0.845 & 0.8 & 0.252 \\
        \midrule
        \end{tabular}
    \caption{Pearson correlation coefficients between swap distances and the
        other ones computed for each used statistical culture.}
    \label{table:pearson_correlation_cultures}
\end{table}

For each dataset we have computed the Pearson Correlation Coefficient
(PCC) between the swap distances and those provided by the other
metrics.  PCC is a classic measure of correlation that takes values
between $-1$ and $1$; its absolute value gives the strength of the
correlation and the sign indicates its positive or negative
nature. \citet{szu-fal-sko-sli-tal:c:map} presented a similar
experiment, but on a much smaller scale, and on a limited set of
metrics.  We present our results in
Table~\ref{table:pearson_correlation}. We see that EMD-positionwise is
most strongly correlated with the swap metric, with a large advantage
over all the metrics, except for pairwise (where the advantage is smaller).
While in \Cref{table:pearson_correlation}, we have provided the Person correlation coefficient between swap distances and the other ones for our $340$ $10\times 50$ elections, in \Cref{table:pearson_correlation_cultures} we take a more fine grained view: 
Recall that the synthetic dataset that we use consists of elections sampled from $13$ synthetic models: for $11$ of them we generated $20$ elections and for the Urn and Mallows model we sampled $60$ elections.
In \Cref{table:pearson_correlation}, for each statistical culture, we give the correlation coefficients between the swap distances of all pairs of elections from this culture and their distances according to our other metrics. 
In \Cref{fig:correlationPlots}, we depict the correlation between swap distances and the other ones on the level of election pairs for the full $340$ elections dataset.

Next, we used the techniques of \citet{szu-fal-sko-sli-tal:c:map} to
draw the maps of elections from the $10 \times 50$ dataset. We show
these maps in Figure~\ref{fig:map} (we omit the discrete metric as its
visualization is nearly meaningless); each dot is an election, its
color corresponds to the statistical model it comes from, and the
points are placed so that their Euclidean distances resemble those
according to a given metric as much as possible (following
\citet{szu-fal-sko-sli-tal:c:map}, we used the algorithm of
\citet{fruchterman1991graph} to place the points).
We also
included the compass elections on the maps (for the swap metric,
their location is approximate).

\begin{figure*}[t]
	\centering
	\begin{subfigure}[b]{0.3\textwidth}
		\includegraphics[width=\textwidth]{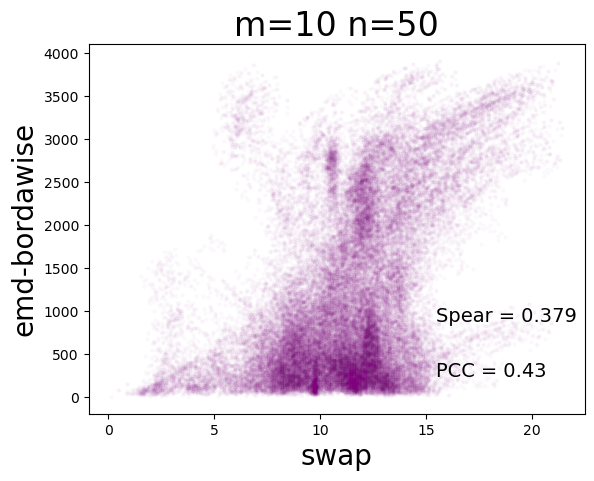}
		\caption{EMD-Bordawise}
	\end{subfigure}\hfill
	\begin{subfigure}[b]{0.3\textwidth}
		\includegraphics[width=\textwidth]{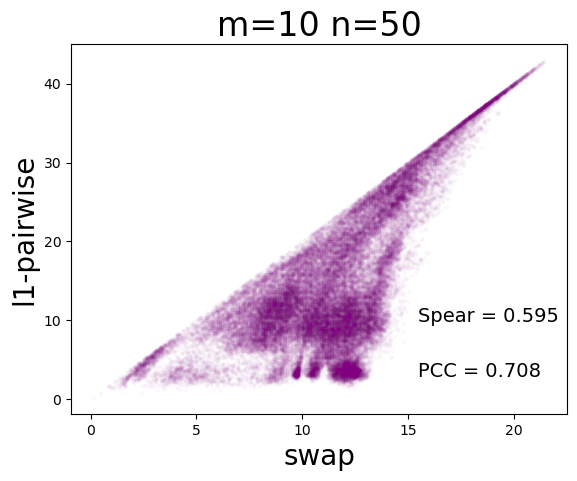}
		\caption{$\ell_1$-Pairwise}
	\end{subfigure}\hfill
	\begin{subfigure}[b]{0.3\textwidth}
		\includegraphics[width=\textwidth]{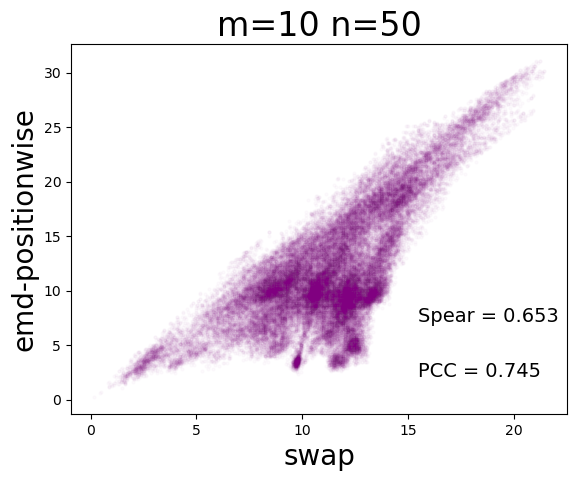}
		\caption{EMD-Positionwise}
	\end{subfigure}%
	\\
	\begin{subfigure}[b]{0.3\textwidth}
		\includegraphics[width=\textwidth]{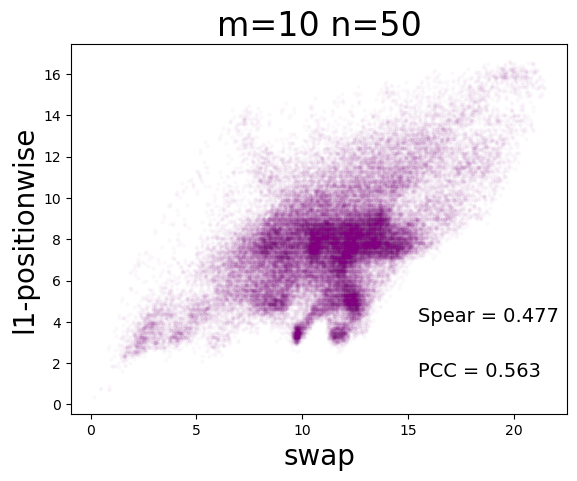}
		\caption{$\ell_1$-Positionwise}
	\end{subfigure}\qquad \qquad
	\begin{subfigure}[b]{0.3\textwidth}
		\includegraphics[width=\textwidth]{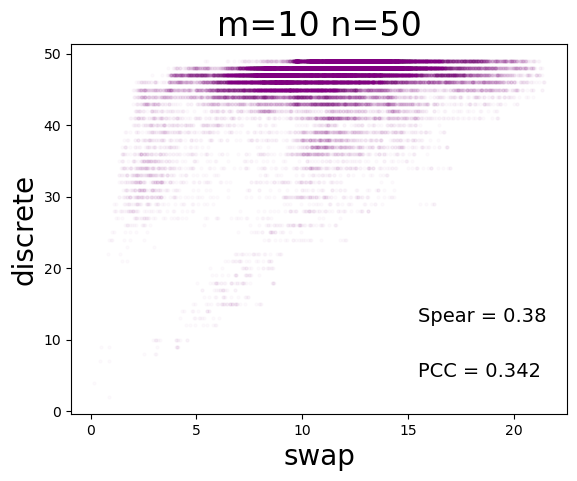}
		\caption{Discrete}
	\end{subfigure}
	\caption{For each of our five non-swap metrics $d$, we show the correlation between the swap distances and $d$ distances in our synthetic dataset of $340$ $10\times 50$ elections. Each point represents a pair of elections and its $x$-coordinate displays the swap distance of this pair and the $y$-coordinate its $d$ distance. Moreover, the Pearson correlation coefficient and the Spearman correlation coefficient between the swap distances and $d$ distances of all pairs are included.} \label{fig:correlationPlots}
\end{figure*}

The maps provided by the isomorphic swap distance and both
positionwise metrics are remarkably similar, but, nonetheless, there
are some differences.
For example, under the swap metric group-separable caterpillar
elections are closer to the IC ones than the group-separable balanced
elections (both on the map and in terms of actual distances), whereas
according to the positionwise metrics this relation is reversed.
Also, $\ell_1$-positionwise
clearly distinguishes between 2D-Sphere and group-separable balanced
elections (like the swap metric), but EMD-positionwise does
not.
Generally, the area between $\UN$ and $\AN$ is quite challenging for
our metrics (fortunately, according to
\citet{boe-bre-fal-nie-szu:c:compass}, only few real-life elections
land there).
The maps for the pairwise and Bordawise metrics
illustrate their flaws identified in the previous sections (e.g.,
Bordawise and pairwise conflate $\UN$ and $\AN$, and the former also
puts all the elections on the diameter, which explains its elongated
shape; the curvature is an artifact of the drawing algorithm).

All in all,  the positionwise metrics seem to perform best in this
section, with the PCC values pointing to the EMD one.

\begin{figure}
	\centering
	\begin{subfigure}[b]{0.32\textwidth}
		\centering
		\includegraphics[width=\textwidth]{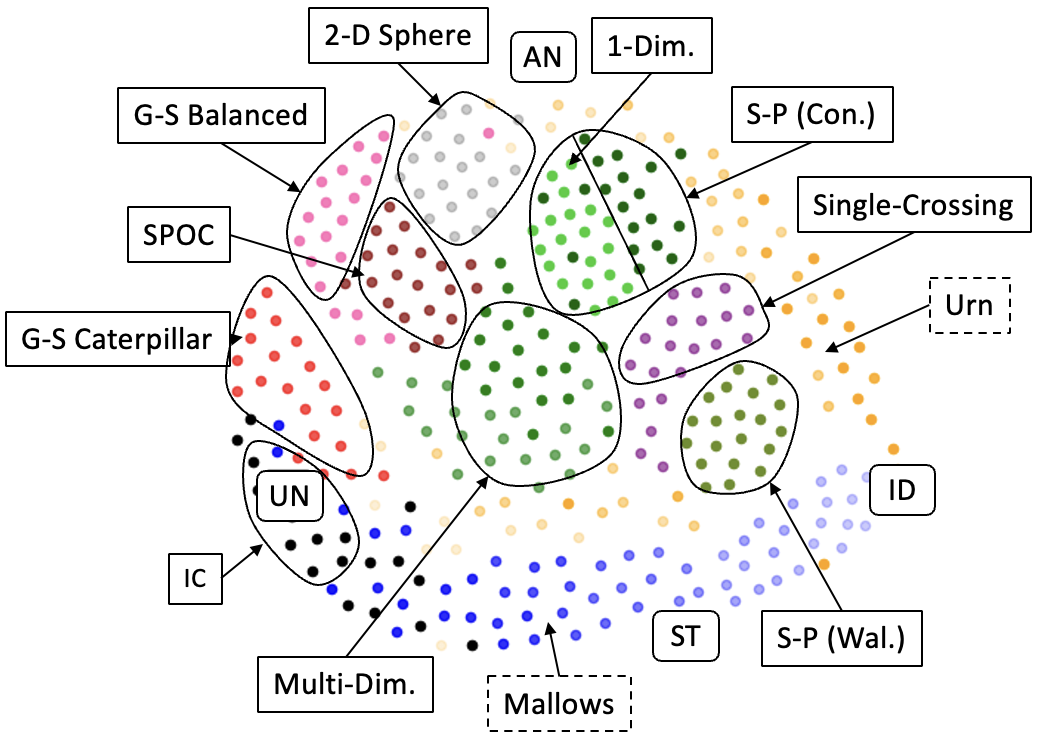}
		\caption{Isomorphic swap}
	\end{subfigure}
	\hfill
	\begin{subfigure}[b]{0.32\textwidth}
		\centering
		\includegraphics[width=\textwidth]{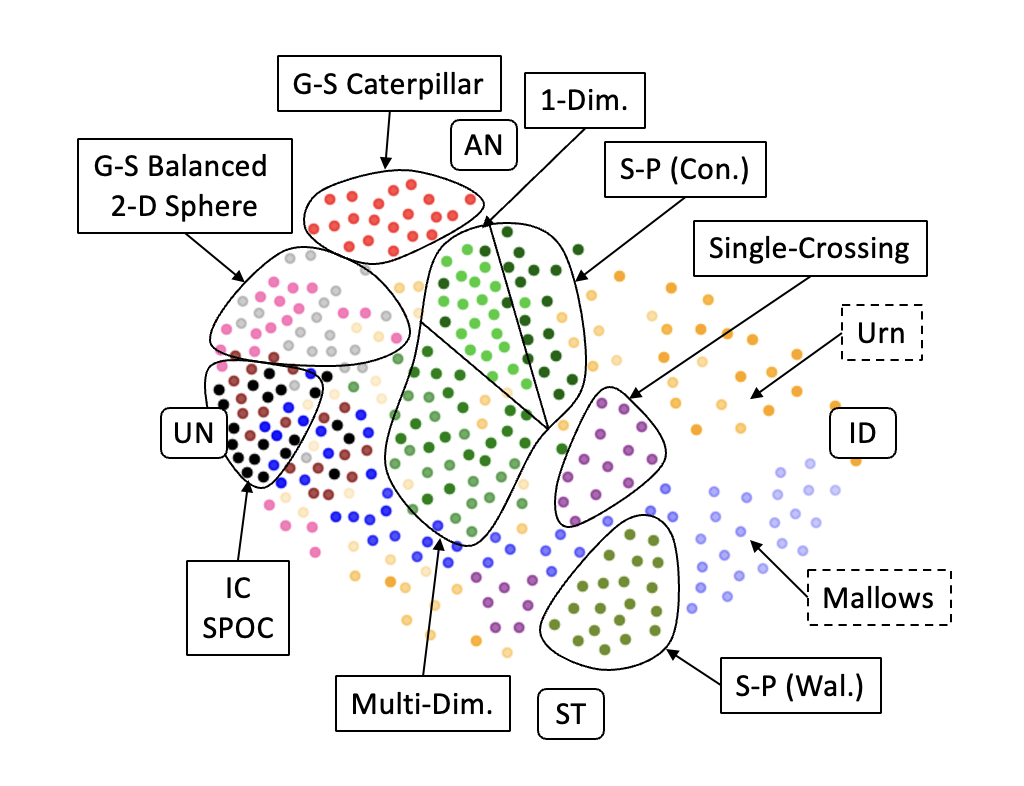}
		\caption{EMD-positionwise}
	\end{subfigure}
	\hfill
	\begin{subfigure}[b]{0.32\textwidth}
		\centering
		\includegraphics[width=\textwidth]{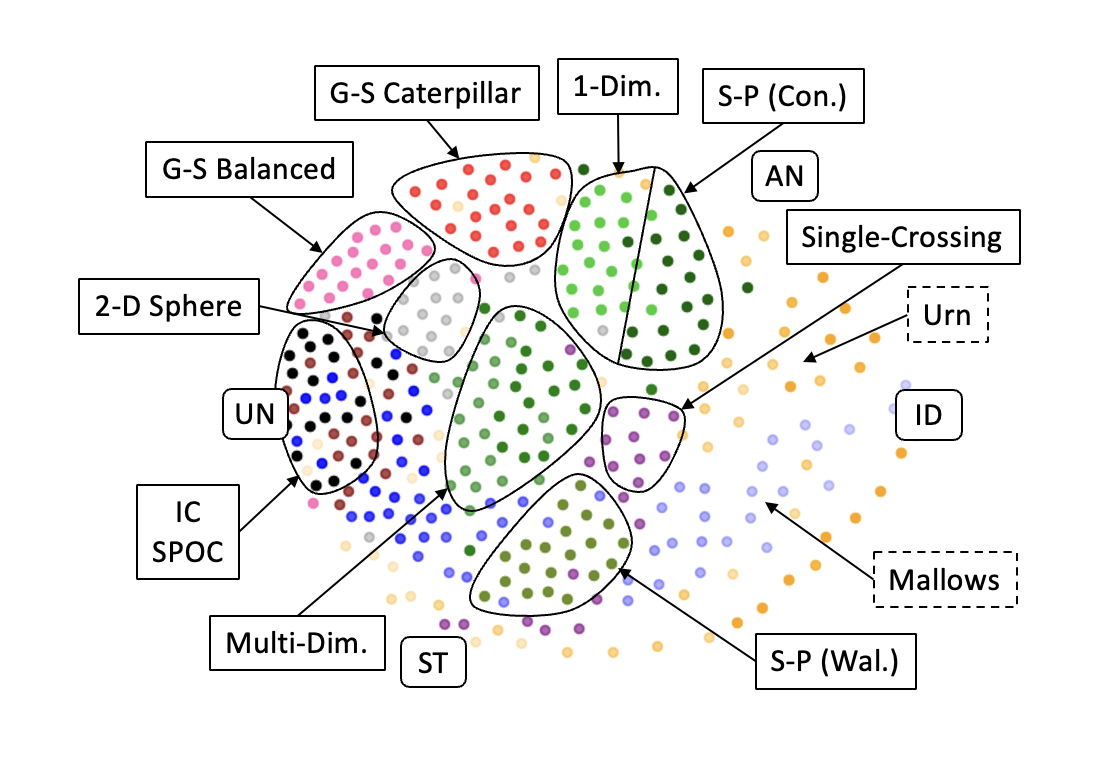}
		\caption{$\ell_1$-positionwise}
	\end{subfigure}
	\hfill
	\begin{subfigure}[b]{0.32\textwidth}
		\centering
		\includegraphics[trim={1.4cm 2.0cm 1.4cm 1cm},clip,width=\textwidth]{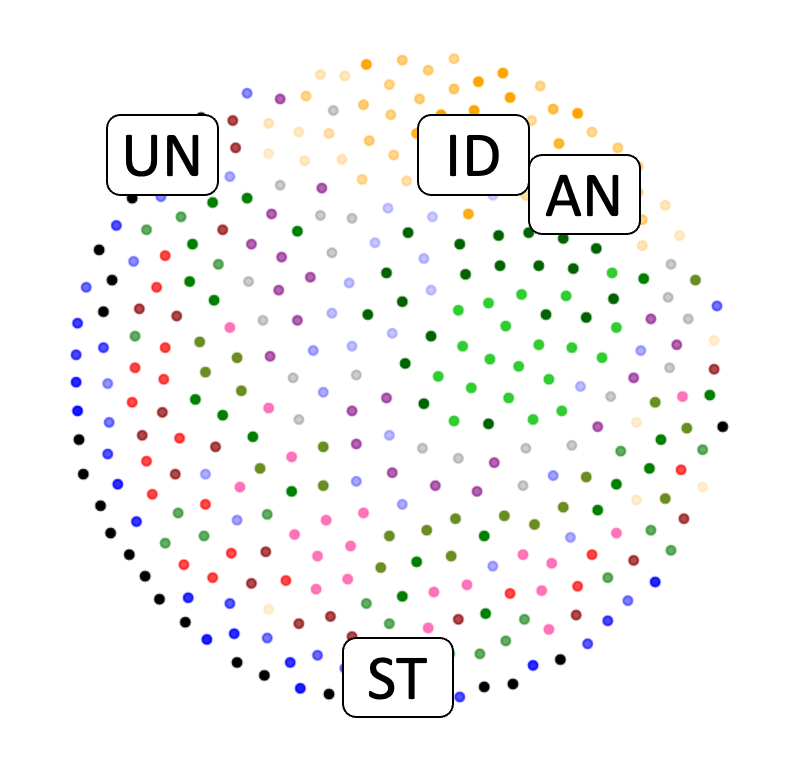}
		\caption{Isomorphic discrete}
	\end{subfigure}
	\hfill
	\begin{subfigure}[b]{0.32\textwidth}
		\centering
		\includegraphics[width=\textwidth]{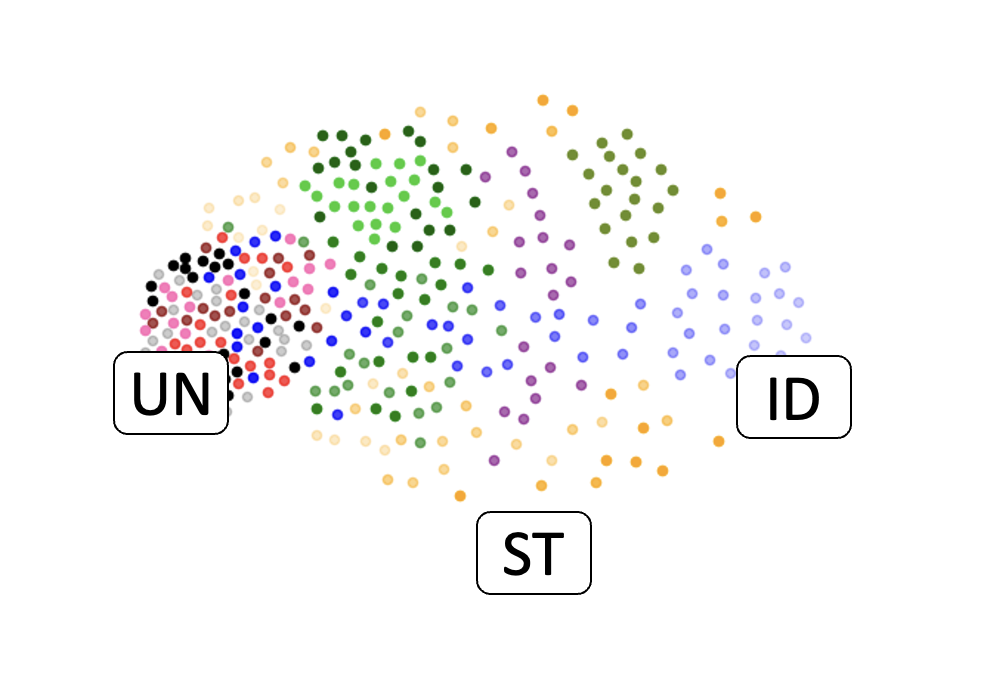}
		\caption{Pairwise}
	\end{subfigure}
	\hfill
	\begin{subfigure}[b]{0.32\textwidth}
		\centering
		\includegraphics[width=\textwidth]{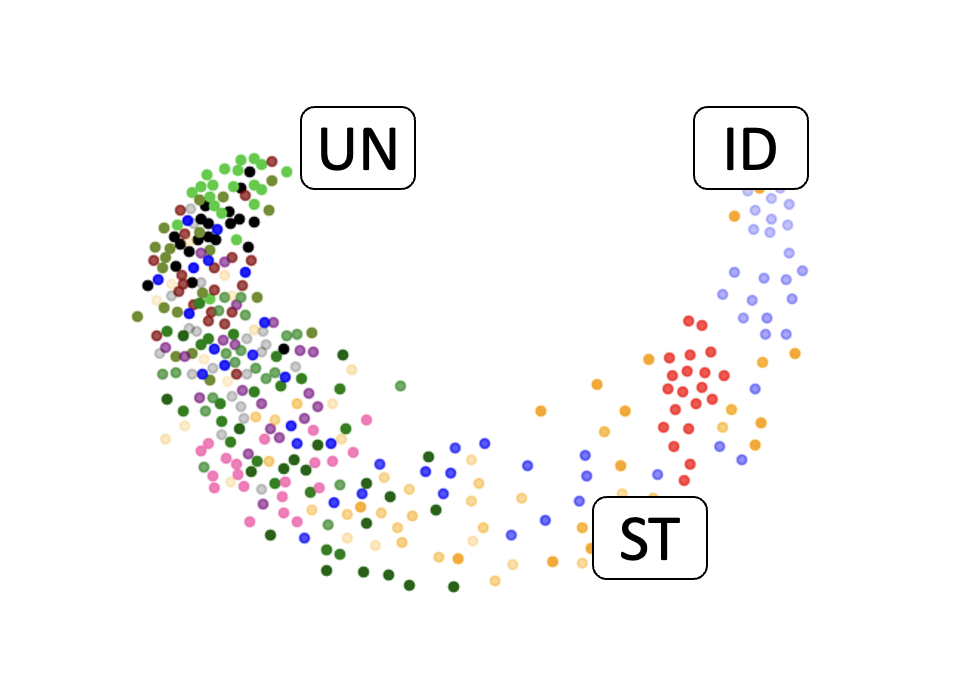}
		\caption{Bordawise}
	\end{subfigure}
	\caption{\label{fig:map}Maps of elections prepared using each of our metrics.}
\end{figure}

\section{Metrics as Graphs} \label{sec:graphs} We conclude by
discussing the intrinsicness of our metrics~(see, e.g., the textbook of
\citet{kha-kir:b:metrics} for more details on this notion).  Consider
a graph whose vertices are equivalence classes of a given metric and
where the edges connect those classes that are at minimum nonzero
distance from each other.
Under the swap distance, each edge corresponds to a swap of adjacent
candidates and each shortest path
corresponds to the distance between its endpoints.  This means that
the swap distance is intrinsic.

\begin{definition}
  Let $\alpha \geq 1$ be a number.  A metric $d$ is $\alpha$-intrinsic
  if for each pair of elections $X$ and $Y$ (with the same number of
  candidates and the same number of voters), there are elections
  $X = E_0, E_1, E_2, \dots, E_{k-1}, E_k = Y$ such that
  $\sum_{i=1}^k d(E_{i-1},E_i) \leq \alpha \cdot d(X,Y)$, and for each
  $i \in [k]$, $d(E_{i-1},E_i)$ is the smallest nonzero distance
  between elections under $d$. If $\alpha = 1$, then $d$ is intrinsic.
\end{definition}
\noindent
An intrinsic metric can be viewed as performing a series of simple,
unit operations.  Among our metrics only swap and discrete are
intrinsic, but the positionwise ones are $2$-intrinsic (this is,
perhaps, the most technically involved of our results).
We defer the rather technically involved proof in \Cref{app:graphs}.

\begin{restatable}{theorem}{intrinsic}
  \label{thm:intrinsic}
  
The swap and discrete metrics are intrinsic, but neither of the
EMD/$\ell_1$-positionwise, pairwise, and Bordawise metric is
intrinsic, but the EMD- and $\ell_1$-positionwise metrics are $2$-intrinsic yet not $\alpha$-intrinsic for any $\alpha<2$.
\end{restatable}

\noindent
We conjecture that pairwise is not $\alpha$-intrinsic for any $\alpha \geq 1$.

\section{Summary}
We found that 
the EMD- and $\ell_1$-positionwise metrics are quite similar to the swap one (which,
computational issues aside, we view as ideal), but the EMD variant
seems better.
This justifies the choice of the EMD-positionwise metric for the maps
of \citet{szu-fal-sko-sli-tal:c:map} and
\citet{boe-bre-fal-nie-szu:c:compass}.  Yet, we ask for a metric that
would perform even better, especially on elections between the uniform
and antagonism ones.
With our study of intrinsicness, we have initiated an axiomatic analysis of our metrics; it would be interesting to extend this analysis to better understand their properties.

\section*{Acknowledgements}
NB was supported by the DFG project MaMu (NI 369/19) and by the DFG project ComSoc-MPMS (NI
369/22).
TW was supported by the Polish National Science Center grant 2018/31/B/ST6/03201.
This project has received funding from the European 
    Research Council (ERC) under the European Union’s Horizon 2020 
    research and innovation programme (grant agreement No 101002854).

\begin{center}
  \includegraphics[width=3cm]{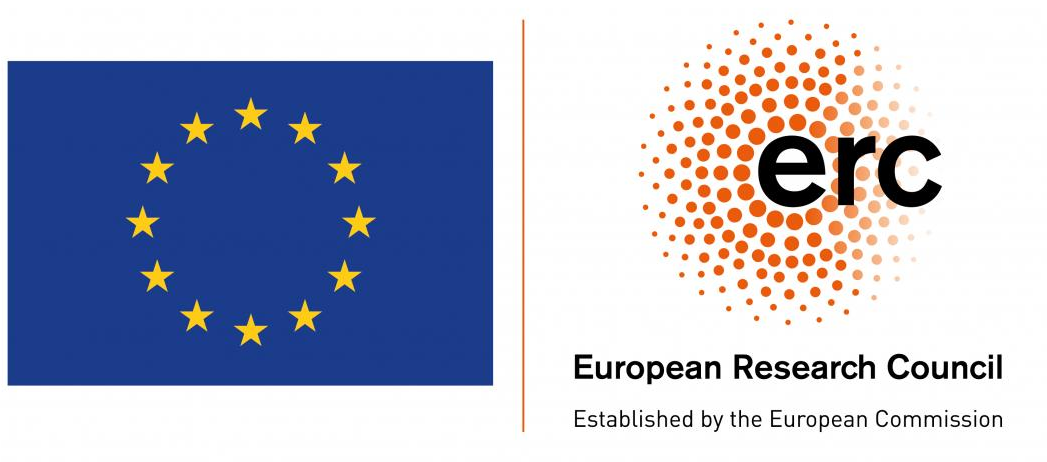}
\end{center}

\clearpage

\bibliographystyle{plainnat}

\clearpage

\appendix

\section[Proof of Theorem 3]{Proof of \Cref{thm:diameter}} \label{app:diameter}

In this section, we prove the following theorem: 
\begin{manualtheorem}{3}\label{th:repeated}
For each two elections $X$ and $Y$ (with an even number $m$ of candidates and  $n = t\cdot m!$ voters where~$t$ is some
positive integer) it holds that $d(X,Y) \leq d(\ID,\UN)$.
\end{manualtheorem}

We do this by separately proving the statement for each of our metrics. 
We first prove it in \Cref{lemma:diam:isomorphic} for the swap and discrete metric. 
Then, in \Cref{lemma:diam:l1-pos} for the $\ell_1$-positionwise metric. 
Subsequently, in \Cref{lemma:diam:emd-pos} for the EMD-positionwise metric, in \Cref{lemma:diam:pair} for the $\ell_1$-pairwise metric, and finally in \Cref{lemma:diam:Borda} for the Bordawise metric. 

\subsection{Swap and Discrete Isomorphic Metrics}

\begin{lemma}
\label{lemma:diam:isomorphic}
For each two elections $X$ and $Y$ with $m$ candidates and $t \cdot m!$ voters it holds that
\begin{align*}
    d_\swap(X,Y) &\le d_\swap(\ID,\UN), \quad \mbox{and}\\
    d_\disc(X,Y) &\le d_\disc(\ID,\UN).
\end{align*}
\end{lemma}
\begin{proof}
Let $d$ be either $d_\swap$ or $d_\disc$.
Consider arbitrary elections $X = (C,V)$ and $Y = (D,U)$ with $|C|=|D|=m$ and $|V|=|U|= n =t \cdot m!$ for some $t \in \mathbb{N}$.
Let us denote $V = \{v_1,v_2,\dots,v_n\}$ and $U = \{ u_1,u_2,\dots,u_n\}$.
From the definition of isomorphic metrics, we get for each bijection $\sigma \in \Pi(C, D)$ the following bound for the distance between $X$ and $Y$:
\[
    d(X,Y) \le \sum_{i=1}^n d(\sigma(v_i),u_i).
\]
Now, if we take the right hand side of this inequality and sum it for all possible bijections $\sigma$, then by rearranging the order of summation we get:
\[
    \sum_{\sigma \in \Pi(C,D)} \sum_{i=1}^n d(\sigma(v_i),u_i) =
    \sum_{i=1}^n \sum_{\sigma \in \Pi(C,D)} d(\sigma(v_i),u_i).
\]
Observe that for each $i \in [n]$,  $\sum_{\sigma \in \Pi(C,D)} d(\sigma(v_i),u_i)$ is the sum of distances between one vote and all possible votes.
Hence, it is exactly the distance between identity and uniformity elections with $m$ candidates and $m!$ votes, which is $t$ times smaller than the distance between the considered elections $\ID$ and $\UN$ with $m$ candidates and $t \cdot m!$ votes.
Thus,
\begin{align*}
    \sum_{\sigma \in \Pi(C,D)} \sum_{i=1}^n d(\sigma(v_i),u_i) &=
    \sum_{i=1}^n \frac{1}{t} \cdot d(\ID,\UN) \\ &=
    \frac{n}{t} \cdot d(\ID,\UN).
\end{align*}
Hence, from pigeonhole principle, there exists a bijection $\sigma^*\in \Pi(C,D)$ such that:
\[
    \sum_{i=1}^n d(\sigma^*(v_i),u_i) \le \frac{n}{t \cdot |\Pi (C,D)|} \cdot d(\ID,\UN) = d(\ID,\UN).
\]
Inserting this into the original bound for the distance between $X$ and $Y$, we get:
\[
    d(X,Y) \le d(\ID,\UN).
\]
\end{proof}

\subsection[L1-Positionwise Metric]{$\boldsymbol{\ell_1}$-Positionwise Metric}

\begin{lemma}
\label{lemma:diam:l1-pos}
For each two elections $X$ and $Y$ with $m$ candidates and $t \cdot m!$ voters it holds that
\[
    \dellpos(X,Y) \le \dellpos(\ID,\UN).
\]
\end{lemma}
\begin{proof}
The $\ell_1$-positionwise metric can be seen as working over matrices whose
entries are nonnegative and whose rows and columns sum up to the same value.
Let $A$ and~$B$ be two such matrices,
whose rows and columns sum up to some value $n'$.
Then, by $A/n'$ and $B/n'$ let us denote the same matrices,
but with their entries divided by $n'$.
Note that $\dellpos(A,B) = n' \cdot \dellpos(A/n',B/n')$.
From now on, we focus on matrices whose entries are nonnegative and whose rows and
columns sum up to $1$ (they are called \emph{bistochastic}).

It is easy to check that $\dellpos(\ID,\UN) = n(2m-2)$.
We prove that for each two $m \times m$ bistochastic matrices $X$ and $Y$
it holds that $\dellpos(X,Y) \le 2m-2$.
For the sake of contradiction, assume otherwise, i.e.,
there exist matrices $X$ and $Y$ such that $\dellpos(X,Y) > 2m-2$.
Let $x_1,\dots,x_m$ be the columns of matrix $X$ and
$y_1,\dots,y_m$ the columns of $Y$.
Without loss of generality, we assume that 
$\dellpos(X,Y) = \sum_{i \in m}\ell_1(x_i,y_i)$;
otherwise we could rearrange the columns of one of the matrices.
For each $i \in [m]$,
we write $x_{i,1},\dots,x_{i,m}$ to denote the entries of column $x_i$;
we use analogous notation for $y_i$.
Then, observe that
\begin{align}
\label{eq:l1-pos:diam:1}
    \dellpos(X,Y) 
        &= \sum_{i,j \in [m]} |x_{i,j} - y_{i,j}| \notag\\
        &= \sum_{i,j \in [m]} \big(\max(x_{i,j},y_{i,j}) - \min(x_{i,j},y_{i,j})\big) \notag\\
        &= \sum_{i,j \in [m]} (x_{i,j} + y_{i,j}) - 2 \sum_{i,j\in [m]} \min (x_{i,j},y_{i,j}) \notag\\
        &= \sum_{j \in [m]} 2 - 2 \sum_{i,j\in [m]} \min (x_{i,j},y_{i,j}) \notag\\
        &= 2m - 2 \sum_{i,j \in [m]} \min (x_{i,j},y_{i,j}).
\end{align}
Thus, if $\dellpos(X,Y) > 2m-2$, then it must hold that:
\begin{equation}
\label{eq:l1-pos:diam:2}
    \sum_{i,j \in [m]} \min (x_{i,j},y_{i,j}) < 1.
\end{equation}

In the following, for each $i,k \in [m]$,
if $i + k >m$ and $i+k$ is used as a column index,
then we take it to be $i + k - m$
(i.e., column indices ``cycle'').
For each $k \in [m]$,
we have $\dellpos(X,Y) \le \sum_{i \in [m]} \ell_1 (x_i,y_{i+k})$;
if this was not the case,
then our assumption that
$\dellpos(X,Y) = \sum_{i \in [m]} \ell_1 (x_i,y_i)$
would have been false.
Consequently, for every $k \in [m]$, repeating the reasoning from Eq.~\eqref{eq:l1-pos:diam:1} and Eq.~\eqref{eq:l1-pos:diam:2}, we get:
\[
    \sum_{i,j \in [m]} \min (x_{i,j},y_{i+k,j}) < 1.
\]
Observe that if $a, b \in [0,1]$,
then it holds that $a \cdot b \le \min(a,b)$.
Since for each $i,j \in [m]$,
we have $x_{i,j},y_{i,j} \in [0,1]$,
for each $k \in [m]$,
it holds that:
\[
    \sum_{i,j \in [m]} x_{i,j} \cdot y_{i + k,j} < 1.
\]
By summing this inequality sidewise for all $k \in \{1,\dots,m\}$,
we obtain:
\[
    \sum_{i,j \in [m]} x_{i,j} \cdot \sum_{k=1}^m y_{k,j} < m.
\]
Since $Y$ is bistochastic, we get that $\sum_{k=1}^m y_{k,j} = 1$.
Thus,
\[
    \sum_{i,j \in [m]} x_{i,j} < m,
\]
which contradicts the fact that $X$ is a bistochastic matrix, as the entries of a bistochastic matrix sum up to $m$.
Hence, for all $m \times m$ bistochastic matrices $X, Y$ we have
$\dellpos(X,Y) \le 2m -2$.
Thus, for all elections with $m$ candidates and $n = t \cdot m!$ voters,
their $\ell_1$-positionwise distance is at most $2n(m-1)$.
\end{proof}

\subsection[EMD-Positionwise Metric]{EMD-Positionwise Metric}
\begin{lemma}
\label{lemma:diam:emd-pos}
For each two elections $X$ and $Y$ with $m$ candidates and $t \cdot m!$ voters it holds that
\[
    \demdpos(X,Y) \le \demdpos(\ID,\UN).
\]
\end{lemma}
\begin{proof}
  The EMD-positionwise metric can be seen as working over matrices whose
  entries are nonnegative and whose rows and columns sum up to the
  same value.
  Let $A$ and~$B$ be two such matrices, whose rows and columns sum up
  to some value $n'$. Let $A/n'$ and $B/n'$ be the same matrices, but
  with their entries divided by $n'$.  Note that
  $\demdpos(A,B) = n' \cdot \demdpos(A/n',B/n')$.  From now on, we
  focus on matrices whose entries are nonnegative and whose rows and
  columns sum up to $1$ (they are called \emph{bistochastic}).

  \citet{boe-bre-fal-nie-szu:c:compass} have shown that
  $\demdpos(\ID,\UN) = \nicefrac{n(m^2-1)}{3}$.  We claim that for
  each two $m \times m$ bistochastic matrices
  $X$ and $Y$ it holds that $\demdpos(X,Y) \leq (m^2 -1)/3$. For the
  sake of contradiction, assume that the opposite holds.  Let
  $x_1, \ldots, x_m$ be the columns of $X$ and $y_1, \ldots, y_m$ be
  the columns of $Y$. Without loss of generality, we assume that
  $\demdpos(X,Y) = \sum_{i \in [m]} \emd(x_i,y_i)$; otherwise we could
  reorder the columns of one of the matrices.  By definition of the
  EMD metric, we have that
  $\sum_{i \in [m]} \emd(x_i,y_i) \textstyle = \sum_{i \in [m]}
  \ell_1(\hat{x}_i,\hat{y}_i)$.  For each $i \in [m]$, we write
  $\hat{x}_{i,1}, \ldots, \hat{x}_{i,m}$ to denote the entries of the
  cumulative vector $\hat{x}_i$; we use analogous notation for
  $\hat{y}_i$.
  Note that in the matrices with columns
  $\hat{x}_1, \ldots, \hat{x}_m$ and $\hat{y}_1, \ldots, \hat{y}_m$,
  for each $j \in [m]$, the $j$-th row sums up to $j$ (we refer to
  this as the \emph{cumulative rows} property).  Using these
  observations, we note that:
  \begin{align}
  \label{eq:emd-pos:diam:1}
    \sum_{i \in [m]} &\emd(x_i,y_i)
    = \sum_{i \in [m]} \ell_1(\hat{x}_i,\hat{y}_i) \notag\\ 
    &= \sum_{i,j \in [m]} |\hat{x}_{i,j} - \hat{y}_{i,j}| \notag\\
    &= \sum_{i,j \in [m]} \big(\max(\hat{x}_{i,j},\hat{y}_{i,j}) - \min(\hat{x}_{i,j},\hat{y}_{i,j})\big) \notag\\
    &= \sum_{i,j \in [m]} (\hat{x}_{i,j} + \hat{y}_{i,j}) - 2 \sum_{i,j \in [m]} \min (\hat{x}_{i,j},\hat{y}_{i,j}) \notag\\
    &= 2 \left( \sum_{j \in [m]} j \right) - 2 \sum_{i,j \in [m]} \min (\hat{x}_{i,j},\hat{y}_{i,j}) \notag\\
    &= m(m+1) - 2 \sum_{i,j \in [m]} \min (\hat{x}_{i,j},\hat{y}_{i,j}).
  \end{align}
  Thus, if $\demdpos(X,Y) > (m^2-1)/3$, then it must hold that:
  \begin{align}
  \label{eq:emd-pos:diam:2}
    \sum_{i,j \in [m]} &\min (\hat{x}_{i,j},\hat{y}_{i,j})  < 
    \frac{1}{2} \left( m(m+1) - \frac{m^2-1}{3}\right) \notag\\
    & = (2m^2 + 3m +1)/6 = (2m+1)(m+1)/6.
  \end{align}
  
  In the following, for each $i,k \in [m]$ with $i+k > m$, if $i+k$ is used as a column
  index, then we take it to be $i+k-m$ (i.e., column indices
  ``cycle'').
  For each $k \in [m]$, we have
  $\demdpos(X,Y) \leq \sum_{i \in [m]} \emd(x_i,y_{i+k})$; if this
  were not the case, then our assumption that
  $\demdpos(X,Y) = \sum_{i \in [m]} \emd(x_i,y_i)$ would have been
  false.
  Consequently, repeating the reasoning from Eq.~\eqref{eq:emd-pos:diam:1} and Eq.~\eqref{eq:emd-pos:diam:2} for every $k\in [m]$, we get: 
  \[
    \textstyle
    \sum_{i,j \in [m]} \min(\hat{x}_{i,j}, \hat{y}_{i+k,j}) < (2m+1)(m+1)/6,
  \]
  If $a, b \in [0,1]$ then $a \cdot b \le \min(a,b)$.  As for each
  $i, j \in [m]$ we have $\hat{x}_{i,j},\hat{y}_{i,j} \in [0,1]$, for
  each $k \in [m]$ we have:
  \[
    \textstyle \sum_{i,j \in [m]} \hat{x}_{i,j} \cdot \hat{y}_{i+k,j} <
    (2m+1)(m +1)/6.
  \]
  By summing this inequality sidewise for all $k \in [m]$, we get:
  \[
    \textstyle
    \sum_{i,j \in [m]} \hat{x}_{i,j} \cdot \sum_{k \in [m]} \hat{y}_{k,j} < m(2m+1)(m+1)/6.
  \]
  By applying the cumulative rows property, we obtain:
  \[
    \textstyle \sum_{j=1}^m j^2 < m(2m+1)(m+1)/6.
  \]
  Since we know that $\sum_{j=1}^m j^2 = {m(2m+1)(m+1)}/{6}$, this is
  a contradiction. Hence, for all $m \times m$ bistochastic matrices
  $X$, $Y$ we have $\demdpos(X,Y) \leq (m^2 -1)/3$. Thus, for all
  elections with $m$ candidates and $n$ voters, their EMD-positionwise
  distance is at most $n(m^2-1)/3$.
\end{proof}

\subsection{Pairwise Metric}

\begin{lemma}
\label{lemma:diam:pair}
For each two elections $X$ and $Y$ with $m$ candidates and $t \cdot m!$ voters it holds that
\[
    d_\pair(X,Y) \le d_\pair(\ID,\UN).
\]
\end{lemma}
\begin{proof}
The pairwise metric can be seen
as working over matrices whose entries are nonnegative,
the entries on the diagonal are zero, and
the sum of entries symmetrical with respect to the diagonal is constant, i.e.,
there is some $n'$ such that $x_{i,j} + x_{j,i} = n'$,
for each $i,j \in [m]$, $i \neq j$.
Let $A$ and $B$ be two such matrices and
by $A/n'$ and $B/n'$ we denote the same matrices but 
with entries divided by $n'$.
Note that $d_\pair (A, B) = n' \!\cdot  d_\pair(A / n' \! , B/n')$.
From now on, we focus on matrices with nonnegative entries,
zeros on the diagonal,
and s.t. entries symmetrical with respect to the diagonal sum up to 1
(let us call them \emph{normalized pairwise} matrices).

Observe that $d_\pair(\ID,\UN) = n(m^2-m)/2$.
We show that for each two $m \times m$ normalized pairwise matrices $X$ and $Y$ it holds that $d_\pair(X,Y) \le (m^2 - m)/2$.
For the sake of contradiction, assume that there exist two normalized pairwise matrices $X$ and $Y$ such that $d_\pair(X,Y) > (m^2 - m)/2$.
Let $x_1,\dots,x_m$ be the columns of matrix $X$ and $y_1,\dots,y_m$ the columns of $Y$.
Also, for each $i \in [m]$, let $x_{i,1},\dots,x_{i,m}$ be the entries of column $x_i$; we use analogous notation for $y_i$ and the columns of matrices introduced later on.
Without loss of generality, we assume that $d_\pair(X,Y) = \sum_{i,j \in [m], i\neq j} |x_{i,j} - y_{i,j}|$;
otherwise we can change the order of columns and their corresponding rows.
Thus, we have:
\begin{align}
\label{eq:pair:diam:1}
    d_\pair(X,Y) &= \sum_{\substack{i, j \in [m] \\ i \neq j}} | x_{i,j} - y_{i,j}| \notag\\ 
           &= \sum_{\substack{i, j \in [m] \\ i \neq j}} \big( \max(x_{i,j},y_{i,j}) - \min(x_{i,j},y_{i,j}) \big) \notag\\
           &= \sum_{\substack{i, j \in [m] \\ i \neq j}} (x_{i,j} + y_{i,j}) - 
           2 \sum_{\substack{i, j \in [m] \\ i \neq j}} \min(x_{i,j},y_{i,j}) \notag\\
           &= m(m-1) - 2 \sum_{\substack{i, j \in [m] \\ i \neq j}} \min(x_{i,j},y_{i,j}).
\end{align}
Hence, if $d_\pair(X,Y) > (m^2-m)/2$, we get that:
\begin{align}
\label{eq:pair:diam:2}
    \sum_{\substack{i, j \in [m] \\ i \neq j}} \min(x_{i,j},y_{i,j})
    &<  \frac{1}{2}(m(m-1) - (m^2 -m)/2) \notag \\
    &=  m(m - 1)/4.
\end{align}
Now, consider matrix $Y'$ with columns $y'_1,\dots,y'_m$ given by
$y'_{i,j} = y_{j,i}$ (i.e., matrix $Y'$ is a transposition of matrix $Y$).
Observe that $Y'$ is still a normalized pairwise matrix and it corresponds to the same elections as $Y$, but with the ordering of candidates reversed.
Hence, it must hold that
\(
    d_\pair(X,Y) \le \sum_{i,j \in [m], i \neq j}|x_{i,j} - y'_{i,j}|;
\)
otherwise our assumption that $d_\pair(X,Y) = \sum_{i,j \in [m], i \neq j}|x_{i,j} - y_{i,j}|$ would have been false.
Thus, repeating the reasoning from Eq.~\eqref{eq:pair:diam:1} and Eq.~\eqref{eq:pair:diam:2} we get:
\[
    \sum_{\substack{i, j \in [m] \\ i \neq j}} \min(x_{i,j},y'_{i,j}) < m(m-1)/4.
\]
Recall that $y'_{i,j} = y_{j,i}$. Thus, adding this inequality sidewise to Eq.~\eqref{eq:pair:diam:2} we get that:
\[
    \sum_{\substack{i, j \in [m] \\ i \neq j}} (\min(x_{i,j},y_{j,i}) + \min(x_{i,j},y_{i,j})) < m(m-1)/2.
\]
If $a, b \in [0,1]$ then $a \cdot b \le \min(a,b)$.
As for each
$i, j \in [m]$ we have $x_{i,j},y_{i,j} \in [0,1]$, we obtain:
\[
    \sum_{\substack{i, j \in [m] \\ i \neq j}} x_{i,j} \cdot (y_{j,i} + y_{i,j}) < m(m-1)/2.
\]
$Y$ is a normalized pairwise matrix, hence $y_{j,i} + y_{i,j}=1$ for every $i \neq j$ and we get that: 
\[
    \sum_{\substack{i, j \in [m] \\ i \neq j}} x_{i,j} < m(m-1)/2.
\]
Since $X$ is also a normalized pairwise matrix, we know that
$\sum_{i,j \in [m],i \neq j} x_{i,j} = m(m-1)/2$.
Thus, we arrive at a contradiction.
Hence, for all $m \times m$ normalized pairwise matrices $X,Y$ we have
$d_\pair (X,Y) \le (m^2 -m)/2$.
Therefore, for all elections with $m$ candidates and $n$ voters, their pairwise distance is at most $n(m^2-m)/2$.
\end{proof}

\subsection{Bordawise Distance}
\begin{lemma}
\label{lemma:diam:Borda}
For each two elections $X$ and $Y$ with $m$ candidates and $n = t \cdot m!$ voters it holds that:
\[
    d_\BOR(X,Y) \leq d_\BOR(\ID,\UN).
\]
\end{lemma}
\begin{proof}
  We first show that for every two elections $E$ and $E'$, such that $E'$ can be obtained
  from $E$ by a single swap of adjacent candidates in a single vote, it holds that:
  \begin{align}\nonumber
    d_\BOR(E,\ID) &+ d_\BOR(E,\UN) =\\ &d_\BOR(E',\ID) + d_\BOR(E',\UN).
      \label{eq:borda-calc}              
  \end{align}
  In other words, we will show that the sum of the distances of an
  election from ID and UN is constant under the Bordawise metric (and
  equal to $d_\BOR(\ID,\UN)$ as every election can be obtained from
  $\ID$ by sufficiently many swaps). Let the sorted Borda score vector
  of election $E$ be $z = (z_1, \ldots, z_m)$.  Without loss of
  generality, we assume that there are numbers $i, j \in [m]$,
  $i < j$, such that the sorted Borda score vector of~$E'$ is
  $z' = (z'_1, \ldots, z'_m)$, where (a) $z'_i = z_i + 1$, (b)
  $z'_j = z_j - 1$, (c) for each $k \in [m] \setminus \{i,j\}$,
  $z'_k = z_k$.  Indeed, a single swap of adjacent candidates can only
  increase the score of one candidate by one point and decrease the
  score of another one by the same value (note that we can swap the
  roles of $E$ and $E'$, to ensure that the assumption that $i < j$ is
  correct).  Let $\hat{z} = (\hat{z}_1, \ldots, \hat{z}_m)$,
  $\hat{z}' = (\hat{z}'_1, \ldots, \hat{z}'_m)$ be the prefix-sum
  variants of vectors $z$ and $z'$.  For each
  $t \in [m] \setminus \{i, \ldots, j-1\}$, we have that
  $\hat{z}'_t = \hat{z}_t$, and for each $t \in \{i, \ldots, j-1\}$, we
  have that $\hat{z}'_t = \hat{z}_t+1$.
  
  Let us now consider the value $d_\BOR(E,\UN) - d_\BOR(E',\UN)$. The
  sorted Borda score vector of $\UN$ is $u = (u_1, \ldots, u_m)$,
  where for each $t \in [m]$, $u_t = \frac{1}{2}n(m-1)$.  Let
  $\hat{u} = (\hat{u}_1, \ldots, \hat{u}_m)$ be its prefix-sum
  variant. For each $t \in [m]$, we have that
  $\hat{u}_t = \frac{1}{2}tn(m-1)$. Since vectors $z$, $z'$, and $u$
  are sorted non-increasingly and sum up to the same value, by a simple counting
  argument we see that for each $t \in [m]$, we have
  $\hat{z}_t \geq \hat{u}_t$ and $\hat{z}'_t \geq \hat{u}_t$.  Using
  this fact, the relation between $\hat{z}$ and $\hat{z}'$, and the
  definition of the EMD distance, we make the following calculations:
  \begin{align*}
    d_\BOR(E,\UN) &- d_\BOR(E',\UN)  = \emd(z,u) - \emd(z',u) \\
      &= \ell_1(\hat{z},\hat{u}) -  \ell_1(\hat{z}',\hat{u}) \\
      &= \textstyle \sum_{t=1}^m |\hat{z}_t-\hat{u}_t| - \sum_{t=1}^m |\hat{z}'_t-\hat{u}_t|\\
      &= \textstyle \sum_{t=i}^{j-1} |\hat{z}_t-\hat{u}_t| - \sum_{t=i}^{j-1} |\hat{z}'_t-\hat{u}_t|\\
      &= \textstyle \sum_{t=i}^{j-1} \big( |\hat{z}_t-\hat{u}_t| - |\hat{z}'_t-\hat{u}_t|\big)\\
      &= \textstyle \sum_{t=i}^{j-1} \big( |\hat{z}_t-\hat{u}_t| - |\hat{z}_t+1-\hat{u}_t|\big)\\
      &= i-j.
  \end{align*}

  Next, we calculate $d_\BOR(E,\ID) - d_\BOR(E',\ID)$. Let the sorted
  Borda score vector of $\ID$ be $v = (v_1, \ldots, v_m)$ and let its
  prefix sum variant be $\hat{v} = (\hat{v}_1, \ldots,
  \hat{v}_m)$. For each $t \in [m]$, we have that $v_t =
  n(m-t)$. Further, for each $t \in [m]$ we have that
  $\hat{v}_t \geq \hat{z}_t$ and $\hat{v}_t \geq \hat{z}'_t$ (too see
  this, observe that $\hat{v}_t$ is a sum of the $t$ highest possible
  Borda scores, multiplied by $n$). Thus we have the following
  calculations (very similar to the previous ones):
  \begin{align*}
    d_\BOR(E,\ID) &- d_\BOR(E',\ID)  = \emd(z,v) - \emd(z',v) \\
      &= \ell_1(\hat{z},\hat{v}) - \ell_1(\hat{z}',\hat{v}) \\
      &= \textstyle \sum_{t=1}^m |\hat{z}_t-\hat{v}_t| - \sum_{t=1}^m |\hat{z}'_t-\hat{v}_t|\\
      &= \textstyle \sum_{t=i}^{j-1} \big( |\hat{z}_t-\hat{v}_t| - |\hat{z}_t+1-\hat{v}_t|\big)\\
      &= j-i.
  \end{align*}
  (The final equality is subtle. It follows by noting that for
  $t \in \{i, \ldots, j-1\}$, it must be that $\hat{z}_t < \hat{v}_t$,
  which, itself, follows from the fact that
  $\hat{z}_t+1 = \hat{z}'_t \leq \hat{v}_t$.)
  Taken together, the above calculations show that:
  \begin{align*}
    d_\BOR(E,\ID)  &- d_\BOR(E',\ID) =\\ &d_\BOR(E',\UN) - d_\BOR(E,\UN),
  \end{align*}
  which is equivalent to Eq.~\eqref{eq:borda-calc}, and which shows
  that the sum of the distances from a given election to ID and UN is
  equal to $d_\BOR(\ID,\UN)$.

  Finally, we show that this means that for every two elections $X$
  and $Y$ we have $d_\BOR(X,Y) \leq d_\BOR(\ID,\UN)$. By the previous
  reasoning, we have that:
  \begin{align*}
    d_\BOR(\ID, X) + d_\BOR(X,\UN) & = d_\BOR(\ID,\UN), \text{ and}\\
    d_\BOR(\ID, Y) + d_\BOR(Y,\UN) & = d_\BOR(\ID,\UN). \numberthis \label{eq:borda}
  \end{align*}
  We claim that one of the following two inequalities  holds:
  \begin{align*}
    d_\BOR(\ID, X) + d_\BOR(\ID, Y) & \leq d_\BOR(\ID,\UN), \text{ or} \\
    d_\BOR(\UN, X) + d_\BOR(\UN, Y) & \leq d_\BOR(\ID,\UN). \numberthis \label{eq:borda2}
  \end{align*}
 If both inequalities do not hold, then we can sum them up arriving at:
  \begin{align*}
    d_\BOR(\ID, X) &+ d_\BOR(\ID, Y) +  d_\BOR(\UN, X) + \\
     &d_\BOR(\UN, Y)  > 2\cdot d_\BOR(\ID,\UN),
  \end{align*}
  which contradicts Eq.~\ref{eq:borda}. 
  Thus, one of the inequalities from Eq.~\ref{eq:borda2} holds and by triangle inequality, we
  have $d_\BOR(X,Y) \leq d_\BOR(\ID,\UN)$. 
\end{proof}
    
\medskip 

From \Cref{lemma:diam:isomorphic,lemma:diam:l1-pos,lemma:diam:emd-pos,lemma:diam:pair,lemma:diam:Borda}, \Cref{th:repeated} follows. 

\section[Proof of Theorem 4]{Proof of \Cref{sec:graphs}} \label{app:graphs}
In this section, we prove the following theorem: 
\intrinsic*

We split the proof in three sections. 
In \Cref{sub:int-disswap}, we prove that swap and discrete metrics are intrinsic. 
In \Cref{sub:int-pairBord}, we show that pairwise and Bordawise metrics are not intrinsic. 
Finally in \Cref{sub:int-pos}, we prove that $\ell_1$- and EMD-positionwise are both $2$-intrinsic but not $\alpha$-intrinsic for any $\alpha<2$ (which, in particular proves that both violate intrinsicness).  

\subsection{Discrete and Swap}\label{sub:int-disswap}
\begin{lemma} \label{le:intrin}
  Swap and discrete metrics are intrinsic.
\end{lemma}
\begin{proof}
  For the swap metric, by definition of the metric, there is a sequence of $k$ swaps of adjacent candidates in some votes that transforms $E_s$ into an election that is isomorphic to $E_t$. With $E_i$ being the election that arises after performing the first $i$ swaps, intrinsicness follows. 
  
  For the discrete metric, given $E_s=(C_s,V_s=(v_1,\dots v_n))$ and $E_t=(C_t,V_t=(v'_1,\dots v'_n))$, let $\sigma\in \Pi (C_s,C_t)$ be the candidate and $\rho\in S_n$ be the voter mapping witnessing the discrete distance between $E_s$ and $E_t$ and let $v_{i_1}, \dots, v_{i_k}\in V_s$ be the votes that contribute one to the discrete distance between $E_s$ and $E_t$ under $\sigma$ and $\rho$. 
  Further, let $V'_s:=V_s\setminus \{v_{i_1}, \dots, v_{i_k}\}$. For some vote $v_i\in V_s$ let $\tau(v_i)$ be the vote $v'_{\rho(i)}$ where each candidate $c\in C_t$ is replaced by $\sigma^{-1}(c)\in C_s$.
  For each $j\in [k-1]$, to construct $E_{j}$, we add all votes from $V'_s$ and $v_{i_j+1}, \dots, v_{i_k}$ and the votes $\tau(v_{i_1}), \dots, \tau(v_{i_j})$.
  Note that using $\sigma$ and $\rho$ as the candidate and voter mapping each two subsequent elections from  $E_s,E_1, E_2, \dots, E_{k-1}, E_t$ are clearly at distance at most one from each other, and, in fact by triangle inequality at distance exactly 1. Form this, intrinsicness follows.
\end{proof}

\subsection{Pairwise and Bordawise}\label{sub:int-pairBord}

\begin{lemma} \label{le:not-intrinsic}
Pairwise and Bordawise metrics are not intrinsic.
\end{lemma}
\begin{proof}
  For Bordawise, consider elections $E_s=(\{a,b,c\}, (a\succ b \succ c, a\succ b \succ c))$ and $E_t=(\{a,b,c\}, (a\succ b \succ c, c\succ b \succ a))$. We have here that $d_\BOR(E_s,E_t)=2$ which is two times the smallest nonzero distance between two elections under Bordawise; however no  $3\times2$ election is at distance $1$ from both $E_s$ and $E_t$: \begin{itemize}
      \item $E_1=(\{a,b,c\}, (a\succ b\succ c, a\succ b\succ c))$: $d_\BOR(E_1,E_t)=2$.
      \item $E_1=(\{a,b,c\}, (a\succ b\succ c, a\succ c\succ b))$: $d_\BOR(E_1,E_t)=3$.
      \item $E_1=(\{a,b,c\}, (a\succ b\succ c, b\succ a\succ c))$: $d_\BOR(E_1,E_t)=3$.
      \item $E_1=(\{a,b,c\}, (a\succ b\succ c, b\succ c\succ a))$: $d_\BOR(E_1,E_s)=2$ 
      \item $E_1=(\{a,b,c\}, (a\succ b\succ c, c\succ b\succ a))$: $d_\BOR(E_1,E_s)=2$.
  \end{itemize}
  However, by intrinsicness, such an election needs to exist, a contradiction to Bordawise being intrinsic.
    
  For pairwise, consider elections $E_s=(\{a,b,c,d,e,f,g\}, (a\succ b \succ c \succ d \succ e\succ f\succ g, e\succ b\succ g\succ d\succ a\succ f\succ c))$ and $E_t=(\{a,b,c,d,e,f,g\}, (a\succ b\succ c\succ g\succ d\succ e\succ f, e\succ b\succ d\succ a\succ f\succ g\succ c))$. 
  It holds that $d_\pair(E_s,E_t)=4$, which is two times the smallest nonzero distance between two elections under pairwise; however we have verified using exhaustive search that no election is at distance $2$ from both $E_s$ and $E_t$ under pairwise (see our code appendix).
\end{proof}

\subsection[L1-Positionwise and EMD-Positionwise]{$\boldsymbol{\ell_1}$-Positionwise and EMD-Positionwise}\label{sub:int-pos}

We start by considering $\ell_1$-positionwise and afterwards examine EMD-positionwise.
However, first let us prove a useful lemma that allows us to consider only paths on elections with the same matching of candidates.
\begin{lemma}\label{lemma:path-matching}
Let $d$ be one of our six metrics.
For any sequence of elections $E_0,E_1,\dots,E_k$, there exists a sequence of elections $E'_0,E'_1,\dots,E'_k$ such that $E_0 = E'_0$, and for every $i \in [k]$ we have that $d(E'_i,E_i)=0$, $d(E'_{i-1},E'_i) = d(E_{i-1},E_i)$, and an optimal matching between candidates in $E'_i$ and $E'_{i-1}$ is the identity.
\end{lemma}
\begin{proof}
Let us show the statement by the induction on $k$.
For $k=0$ the statement is trivial.
Thus, assume that $k>0$ and that the statement holds for $k-1$.

Fix an arbitrary sequence of elections $E_0,E_1,\dots,E_k$.
By the induction assumption, there exists a sequence of elections $E'_0,E'_1,\dots,E'_{k-1}$ such that $E_0 = E'_0$ and for every $i \in [k-1]$ we have that $d(E'_i,E_i)=0$, $d(E'_{i-1},E'_i) = d(E_{i-1},E_i)$, and an optimal matching between candidates in $E'_i$ and $E'_{i-1}$ is the identity.
Since $d(E'_{k-1},E_{k-1})=0$, we have that $d(E_{k-1},E_k) = d(E'_{k-1},E_k)$.
Let $\sigma$ be an optimal matching of candidates between $E'_{k-1}$ and $E_k$.
Then, let us construct election $E'_k$ from $E_k=(C_k,V_k)$ by exchanging each vote $v\in V_k$ by vote $\sigma(v)$, i.e., $E'_k = (C_k,\sigma(V_k))$.
In this way, $d(E'_k,E_k)=0$,
\[
    d(E'_{k-1},E'_k)=d(E'_{k-1},E_k)=d(E_{k-1},E_k),
\]
and the optimal matching between $E'_k$ and $E'_{k-1}$ is identity.
This concludes the proof.
\end{proof}

\subsubsection[L1-Positionwise]{$\boldsymbol{\ell_1}$-Positionwise}

\begin{lemma}
\label{prop:l1-pos:intrinsic-degree}
The $\ell_1$-positionwise metric is 2-intrinsic.
\end{lemma}
\begin{proof}

We need to show that for every two elections $E_s$ and $E_t$ there exist
elections $E_s = E_0,E_1,\dots,E_{k-1},E_k = E_t$ such that
$\sum_{i=1}^k \dellpos(E_{i-1},E_i) \leq 2 \cdot \dellpos(E_s,E_t)$, and for each
$i \in [k]$, $\dellpos(E_{i-1},E_i) = 4$, which is the smallest nonzero distance under the $\ell_1$-positionwise metric.
We prove this by induction on $\dellpos(E_s,E_t)$.
Fix arbitrary elections $E_s$ and $E_t$ with $n$ voters and $m$ candidates.
If $\dellpos(E_s,E_t)=0$ or $\dellpos(E_s,E_t)=4$, the statement trivially follows.
Thus, assume that $\dellpos(E_s,E_t)>4$ and that the statement holds for all elections $\hat{E}_s$ and $\hat{E}_t$ with $\dellpos(\hat{E}_s,\hat{E}_t)<\dellpos(E_s,E_t)$.
Let $X$ with columns $x_1,\dots,x_m$ be the position matrix of $E_s$ and $Y$ with columns $y_1,\dots,y_m$ be the position matrix of $E_t$.
Without loss of generality, let us assume that $\dellpos(E_s,E_t)= \sum_{i \in [m]} \ell_1(x_i,y_i)$;
otherwise we could reorder the columns of one of the matrices.

Intuitively, in what follows, we examine matrix $X$ and consider two columns $c,c'$ from $X$, two rows $r,r'$ from $X$, and matrix $X'$ that results from $X$ by, in column $c$, subtracting 1 from row $r$ and adding 1 to row $r'$ and, in column $c'$, adding 1 to row $r$ and subtracting 1 from row $r'$.
As we will show later there always exist such $c,c',r,r'$ that $X'$ is a position matrix of some election $E'_s$, and it holds that $\dellpos(E'_s,E_t) \le \dellpos(E_s,E_t) - 2$ and $\dellpos(E_s,E'_s) = 4$.
Using this, our statement will then follow by induction.
As an illustration of our approach, see the following example:

\begin{example}\label{ex:lemma:l1-pos-2-intrinsic}
  Let the first two of the following matrices be example position matrices $X$ and $Y$ of some elections $E_s$ and $E_t$, respectively (we have $n=5$ voters and $m=4$ candidates).
  Then, the third matrix is a constructed position matrix $X'$ created following our above described approach of some election $E'_s$.
  \begin{align*}
         \small
         \kbordermatrix{ &\! c &\!   &\! c' &\!   \\
      \!\!\!\! & \!        1 &\! 2 &\! 1  &\! 1  \\
    r \!\!\!\! & \!        \textbf{3} &\! 1 &\! \textbf{0}  &\! 1  \\
      \!\!\!\! & \!        1 &\! 2 &\! 2  &\! 0  \\
    r'\!\!\!\! & \!        \textbf{0} &\! 0 &\! \textbf{2}  &\! 3  \\
    }\!\!\!\!&&
         \small
         \kbordermatrix{ &\! c &\!   &\! c' &\!   \\
      \!\!\!\! &\!           1 &\! 2 &\! 2  &\! 0  \\
    r \!\!\!\! &\!           \textbf{2} &\! 2 &\! \textbf{0}  &\! 1  \\
      \!\!\!\! &\!           1 &\! 1 &\! 2  &\! 1  \\
    r'\!\!\!\! &\!           \textbf{1} &\! 0 &\! \textbf{1}  &\! 3  \\
    }\!\!\!\!&&
         \small
         \kbordermatrix{ &\! c &\!   &\! c' &\!   \\
      \!\!\!\! &\!           1 &\! 2 &\! 1  &\! 1  \\
    r \!\!\!\! &\!           \textbf{2} &\! 1 &\! \textbf{1}  &\! 1  \\
      \!\!\!\! &\!           1 &\! 2 &\! 2  &\! 0  \\
    r'\!\!\!\! &\!           \textbf{1} &\! 0 &\! \textbf{1}  &\! 3  \\
    }&&\\
    X \quad\quad && Y \quad\quad && X' \quad\quad &&
  \end{align*}
    
  Observe that $\dellpos(E_s,E_t) = 8$.
  Hence, indeed, we get that $\dellpos(E'_s,E_t) = 6 \le \dellpos(E_s,E_t) - 2$.
  Moreover, we have $\dellpos(E_s,E'_s) = 4$.
  This concludes the example.
\end{example}

For each $i \in [m]$ we write $x_{i,1},\dots,x_{i,m}$ to the denote the entries of $x_i$;
we use analogous notation for $y_i$ and the columns of matrices introduced later on.
Since $\dellpos(E_s,E_t)>4$, there must exist a column, $c \in [m]$, such that $\ell_1(x_c,y_c) > 0$.
Moreover, the sums of entries in both $x_c$ and $y_c$ are equal (to $n$), hence there is an $r \in [m]$ such that
$x_{c,r} > y_{c,r}$ and also $r' \in [m]$ such that $x_{c, r'} < y_{c, r'}$.
Furthermore, the sums of entries in row $r'$ in both matrices are equal (to $n$), thus the fact that $x_{c, r'} < y_{c, r'}$ implies that there exists column $c' \in [m]$ such that $x_{c',r'} > y_{c',r'}$.
Building upon this, let us construct matrix $X'$ with columns $x'_1,\dots,x'_m$ defined as follows:
\[
    x'_{i,j} =
    \begin{cases}
        x_{i,j} - 1, & \mbox{if } (i,j) \in \{(c,r),(c',r')\},\\
        x_{i,j} + 1, & \mbox{if } (i,j) \in \{(c,r'),(c',r)\},\\
        x_{i,j}, & \mbox{otherwise.}
    \end{cases}
\]
Observe that for $(i,j) \in \{(c,r),(c',r')\}$ we have that $x'_{i,j} = x_{i,j} - 1 \ge y_{i,j} \ge 0$.
Thus, $x'_{i,j} \ge 0$ for every $i,j \in [m]$.
Moreover, the sums of entries of each row and column of matrix $X'$ are still equal to $n$.
Thus, as proven by \citet{boe-bre-fal-nie-szu:c:compass}, $X'$ is a position matrix of some election--- which we denote by $E'_s$.

Observe that $X$ and $X'$ differ only on columns $c$ and $c'$.
Also, for column $c$ we have $\ell_1(x'_c,y_c) = \ell_1(x_c,y_c) - 2$.
Moreover, for column $c'$ we get $\ell_1(x'_{c'},y_{c'}) = \ell_1(x_{c'},y_{c'}) - 2$, if $x_{c',r} < y_{c',r}$, and
$\ell_1(x'_{c'},y_{c'}) = \ell_1(x_{c'},y_{c'})$, otherwise.
Since other columns of $X$ and $X'$ are identical, we get that
\[
    \sum_{i \in [m]} \ell_1(x'_i,y_i) =
    \begin{cases}
        \sum_{i \in [m]} \ell_1(x_i,y_i) - 4, & \mbox{if } x_{c',r} < y_{c',r},\\
        \sum_{i \in [m]} \ell_1(x_i,y_i) - 2, & \mbox{otherwise.}
    \end{cases}
\]
Hence, $\dellpos(E'_s,E_t) \le \dellpos(E_s,E_t) - 2$.
By the induction assumption, this means that there exist elections $E'_s = E_0,E_1,\dots,E_{k-1},E_k = E_t$ such that 
\begin{equation}
\label{eq:lemma:l1-pos-2-intrinsic}
    \sum_{i=1}^k \dellpos(E_{i-1},E_i) \le 
    2 \cdot \dellpos(E'_s,E_t) \le
    2 \cdot \dellpos(E'_s,E_t) - 4
\end{equation}
and for each $i \in [k]$, we have that $\dellpos(E_{i-1},E_i) = 4$.

On the other hand, observe that
\[
    \sum_{i \in [m]} \ell_1(x_i,x'_i) = \ell_1(x_c,x'_c) + \ell_1(x_{c'},x'_{c'}) = 4.
\]
Since $E'_s$ is closer to $E_t$ than $E_s$, it is not possible that $\dellpos(E_s,E'_s)=0$.
Thus, $\dellpos(E_s,E'_s)=4$.

Building upon this, let us denote $E_{-1} = E_s$.
Then, by Eq.~\eqref{eq:lemma:l1-pos-2-intrinsic} elections $E_s = E_{-1},E_0,E_1,\dots,E_{k-1},E_k = E_t$ are such that
\begin{align*}
    \sum_{i=0}^k \dellpos(E_{i-1},E_i) &\le 
    2 \cdot \dellpos(E'_s,E_t) - 4 + \dellpos(E'_s,E_s) \\ &=
    2 \cdot \dellpos(E_s,E_t).
\end{align*}
Moreover, for each $i \in \{0,1,\dots,k\}$, we have that $\dellpos(E_{i-1},E_i) = 4$, which concludes the proof.
\end{proof}

\begin{lemma}\label{le:notint-l1}
The $\ell_1$-positionwise metric is not $\alpha$-intrinsic for any $\alpha < 2$.
\end{lemma}
\begin{proof}
Let us consider an identity election, $\ID$ over candidates $c_1,\dots, c_m$, with $n > m > 2$ votes: $c_1 \succ c_2 \succ \dots \succ c_m$.
Also, let us consider election $E$ that is obtained from $\ID$ by exchanging one of the votes by vote $c_2 \succ c_3 \succ \dots \succ c_m \succ c_1$.
Observe that position matrix of $E$ is:
\[
    \begin{bmatrix}
      n\! -\! 1 &  1  &  0  & \cdots &  0  &  0  \\
       0  & n\! -\! 1&  1  & \cdots &  0  &  0  \\
       0  &  0  & n\! -\! 1 & \cdots &  0  &  0  \\
      \vdots&\vdots&\vdots&\ddots&\vdots&\vdots \\
       0  &  0  &  0  & \cdots & n\! -\! 1 &  1  \\
       1  &  0  &  0  & \cdots &  0  & n\! -\! 1 \\
    \end{bmatrix}.
\]
Thus, it is optimal to match $c_i$ in $\ID$ to $c_i$ in $E$ for every $i \in [m]$ and we get that
$\dellpos(E,\ID)=2m$.

Let $A$ be an $m \times m$ matrix.
Let $a_1,\dots,a_m$ be the columns of matrix $A$ and for each $i \in [m]$, let $a_{i,1},\dots,a_{i,m}$ be the entries of $a_i$.
We will say that matrix $A$ is \emph{atomic}, if there exist $c,c',r,r' \in [m]$ such that
\[
    a_{i,j} =
    \begin{cases}
        1, & \mbox{if } (i,j) \in \{(c,r),(c',r')\},\\
        - 1, & \mbox{if } (i,j) \in \{(c,r'),(c',r)\},\\
        0, & \mbox{otherwise.}
    \end{cases}
\]
Observe that for every two elections $F,F'$ with position matrices $X, X'$ having columns $x_1,\dots,x_m$ and $x'_1,\dots,x'_m$ such that
$\dellpos(F,F')=\sum_{i \in [m]}\ell_1(x_i,x'_i)=4$,
it holds that there exists an atomic matrix $A$ such that $X = X' + A$.
Hence, by Lemma~\ref{lemma:path-matching}, finding an intrinsic path of length $k$ between elections $\ID$ and $E$ is equivalent to finding atomic matrices $A_1,\dots,A_k$ such that their sum is equal to
\[
    \begin{bmatrix}
      -1  &  1  &  0  & \cdots &  0  &  0  \\
       0  & -1  &  1  & \cdots &  0  &  0  \\
       0  &  0  & -1  & \cdots &  0  &  0  \\
      \vdots&\vdots&\vdots&\ddots&\vdots&\vdots \\
       0  &  0  &  0  & \cdots & -1  &  1  \\
       1  &  0  &  0  & \cdots &  0  & -1  \\
    \end{bmatrix}.
\]
Observe that the rank of this matrix is equal to $m-1$.
Also, the rank of each atomic matrix is equal to $1$.
Since the rank of a matrix is subadditive, we get that $k \ge m-1$.
From the definition of the intrinsicness degree we obtain that
\(
    \alpha \ge k \cdot d_{min} / \dellpos(E,\ID).
\)
Thus,
\[
    \alpha \ge \frac{4(m-1)}{2m} = 2 - \frac{2}{m}.
\]
Since $m$ can be arbitrarily large, we get that $\alpha \ge 2$.
\end{proof}

\subsubsection{EMD-Positionwise}

\begin{lemma}\label{prop:emd-pos:intrinsic-degree}
EMD-positionwise metric is 2-intrinsic.
\end{lemma}
\begin{proof}
We need to show that for every two elections $E_s$ and $E_t$ there exist
elections $E_s = E_0,E_1,\dots,E_{k-1},E_k = E_t$ such that
$\sum_{i=1}^k \demdpos(E_{i-1},E_i) \leq 2 \cdot \demdpos(E_s,E_t)$, and for each
$i \in [k]$, $\demdpos(E_{i-1},E_i) = 2$, which is the smallest nonzero distance of two elections under the EMD-positionwise metric.
We prove this by induction on $\demdpos(E_s,E_t)$.
For this fix two arbitrary elections $E_s$ and $E_t$ with $n$ voters and $m$ candidates.
If $\demdpos(E_s,E_t)=0$ or $\demdpos(E_s,E_t)=2$, the statement trivially follows.
Thus, assume that $\demdpos(E_s,E_t)>2$ and that the statement holds for all elections $\hat{E}_s$ and $\hat{E}_t$ with $\demdpos(\hat{E}_s,\hat{E}_t) < \demdpos(E_s,E_t)$.
Let $X$ with columns $x_1,\dots,x_m$ be the position matrix of $E_s$ and $Y$ with columns $y_1,\dots,y_m$ the position matrix of $E_t$.
Without loss of generality, let us assume that $\demdpos(E_s,E_t)= \sum_{i \in [m]} \emd(x_i,y_i)$;
otherwise we could reorder columns of one of the matrices.
For each column $i \in [m]$ we write $x_{i,1},\dots,x_{i,m}$ to the denote the entries of $x_i$;
we use analogous notation for $y_i$ and the columns of matrices introduced later on.
Recall also that for any vector $v$, by $\hat{v}$ we denote the vector of its prefix sums.

Intuitively, in what follows, as in proof of Lemma~\ref{prop:l1-pos:intrinsic-degree},
we examine matrix $X$ and consider two columns $c,c'$ from $X$, two rows $r,r'$ from $X$, and matrix $X'$ that results from $X$ by, in column $c$, subtracting 1 from row $r$ and adding 1 to row $r'$ and, in column $c'$, adding 1 to row $r$ and subtracting 1 from row $r'$.
As we will show, it is always possible to choose such rows and columns that $X'$ is still a position matrix of some elections, $E'_s$, and $\demdpos(E'_s,E_t) \le \demdpos(E_s,E_t) - 2 \cdot | r - r'|$.
Furthermore, we will prove that between $E_s$ and $E'_s$ we can find a path of at most $2 \cdot |r - r'|$ elections such that each consecutive two are at EMD-positionwise distance $2$.
Using this, by induction, we will obtain the statement.
See the following example:

\begin{example}\label{ex:lemma:emd-pos-2-intrinsic:1}
  Let the first two of the following matrices be example position matrices $X$ and $Y$ of some elections $E_s$ and $E_t$, respectively (we have $n=6$ voters and $m=4$ candidates).
  Then, the third matrix is a constructed position matrix $X'$ created following our above described approach of some election $E'_s$.
  \begin{align*}
         \small
         \kbordermatrix{ &\! &\! c &\!    &\! c' \\
    r'\!\!\!\! & \!        1 &\! \textbf{1} &\! 2  &\! \textbf{2}  \\
      \!\!\!\! & \!        3 &\! 0 &\! 3  &\! 0  \\
     r\!\!\!\! & \!        0 &\! \textbf{4} &\! 0  &\! \textbf{2}  \\
      \!\!\!\! & \!        2 &\! 1 &\! 1  &\! 2  \\
    }\!\!\!\!&&
         \small
         \kbordermatrix{ &\! &\! c &\!    &\! c' \\
    r'\!\!\!\! & \!        1 &\! \textbf{2} &\! 2  &\! \textbf{1}  \\
      \!\!\!\! & \!        3 &\! 0 &\! 3  &\! 0  \\
     r\!\!\!\! & \!        1 &\! \textbf{3} &\! 0  &\! \textbf{2}  \\
      \!\!\!\! & \!        1 &\! 1 &\! 1  &\! 3  \\
    }\!\!\!\!&&
         \small
         \kbordermatrix{ &\! &\! c &\!    &\! c' \\
    r'\!\!\!\! & \!        1 &\! \textbf{2} &\! 2  &\! \textbf{1}  \\
      \!\!\!\! & \!        3 &\! 0 &\! 3  &\! 0  \\
     r\!\!\!\! & \!        0 &\! \textbf{3} &\! 0  &\! \textbf{3}  \\
      \!\!\!\! & \!        2 &\! 1 &\! 1  &\! 2  \\
    }\!\!\!\!&&\\
    X \quad\quad && Y \quad\quad && X' \quad\quad &&
  \end{align*}
  Observe that we have $\demdpos(E_s,E_t) = 1 + 2 + 0 + 3 = 6$, $\demdpos(E'_s,E_t) = 1 + 0 + 0 + 1 =2$, and $r - r' = 2$.
  So, indeed, $\demdpos(E'_s,E_t) = \demdpos(E_s,E_t) - 2(r-r')$.
  This concludes the example.
\end{example}

We begin by finding the correct columns $c,c'$ and rows $r,r'$.
Our goal is to find them such that $r>r'$, $x_{c,r} > y_{c,r}$ and $x_{c',r'} > y_{c',r'}$, but at the same time $x_{c,i} = y_{c,i}$ and $x_{c',i} = y_{c',i}$ for every $i \in \{r'+1,r'+2,\dots,r-1\}$.
In this way, when we subtract 1 from $x_{c,r}$ and $x_{c',r'}$ and add 1 to $x_{c,r'}$ and $x_{c',r}$, then both $\emd(x_c,y_c)$ and $\emd(x_{c'},y_{c'})$ decrease by exactly $r-r'$
(we will prove it in more detail later on).
Thus, the EMD-positionwise distance to election $E_t$ decreases by at least $2(r-r')$.

We want to take $c,r \in [m]$ such that $x_{c,r} > y_{c,r}$, $\hat{x}_{c,r-1} < \hat{y}_{c,r-1}$, and $r$ is minimal.
But first, let us prove that there always exists at least one pair $c,r \in [m]$ such that $x_{c,r} > y_{c,r}$ and $\hat{x}_{c,j} < \hat{y}_{c,j}$.
To this end, observe that since $\demdpos(E_s,E_t)>2$, there exists a column, $c \in [m]$, such that $\emd(x_c,y_c)>0$.
Since the sums of entries of column $c$ in both matrices $X$ and $Y$ are equal (to $n$), there exists a row, $p \in [m]$, such that $x_{c,p} < y_{c,p}$.
Let us fix $c,p \in [m]$ such that $x_{c,p} < y_{c,p}$ and $p$ is minimal, i.e., for every $i,j \in [m]$ such that $x_{i,j} < y_{i,j}$ we have $j \ge p$.
Then, in particular, we have $x_{c,j} \ge y_{c,j}$, for each $j \in [p-1]$.
Assume that there exists $j \in [p-1]$ such that $x_{c,j} > y_{c,j}$.
Since the sums of entries in row $j$ in matrices $X$ and $Y$ are equal (to $n$), there exists column $i \in [m]$ such that $x_{i,j} < y_{i,j}$.
However, since $j < p$, this contradicts the assumption that $p$ is minimal.
Therefore, we get that $x_{c,j} = y_{c,j}$, for each $j \in [p-1]$.
Thus, $\hat{x}_{c,p} < \hat{y}_{c,p}$, as $x_{c,p} < y_{c,p}$.
Since the sums of entries in column $c$ in both matrices are equal (to $n$), there exists such $r$ that $x_{c,r} > y_{c,r}$.
Let us take the minimal such $r$.
Then, for every $j \in \{p+1,p+2,\dots,r-1\}$ we have that $x_{c,j} \le y_{c,j}$.
Thus, $\hat{x}_{c,r-1} < \hat{y}_{c,r-1}$.
This means that we have found $c,r$ such that $x_{c,r} > y_{c,r}$ and $\hat{x}_{c,r-1} < \hat{y}_{c,'-1}$.

Therefore, let us take $c, r \in [m]$
such that $x_{c,r} > y_{c,r}$, $\hat{x}_{c,r-1} < \hat{y}_{c,r-1}$, and there is no $i, j \in [m]$ such that $j < r$, $x_{i,j} > y_{i,j}$, $\hat{x}_{i,j-1} < \hat{y}_{i,j-1}$, i.e., $r$ is minimal with respect to this property.
Observe from the fact that $\hat{x}_{c,r-1} < \hat{y}_{c,r-1}$ implies that there exists row $q \in [r-1]$ such that $x_{c,q} < y_{c,q}$.
Let us take such $q \in [r-1]$ that is minimal with such property.
Then, for every $j \in \{q+1,q+2,\dots,r-1\}$ we have $x_{c,q} = y_{c,q}$.

Now, let us take $c',r' \in [m]$ such that $x_{c',r'} > y_{c',r'}$, $q \le r' < r$, $c' \neq c$, and $r'$ is maximal with respect to this property.
To see that there exists at least one pair $i,j \in [m]$ such that $x_{i,j} > y_{i,j}$, $q \le j < r$, and $i \neq c$,
observe that the sums of entries in row $q$ in matrices $X$ and $Y$ are equal.
Hence, since $x_{c,r} < y_{c,q}$, there exists column $i \in [m]$ such that $x_{i,q} > y_{i,q}$.
Taking $j=q$, we get that indeed $q \le j < r$ and $i \neq c$.
Hence, there exists $c',r' \in [m]$ such that $x_{c',r'} > y_{c',r'}$, $q \le r' < r$, $c' \neq c$, and $r'$ is maximal.
Observe that this means that in every row $j \in \{r'+1,r'+2,\dots,r-1\}$ we have $x_{c',j} = y_{c',j}$;
otherwise if $x_{c',j} > y_{c',j}$, then $r'$ is not maximal as we assumed ($r' < j < r$ and $c' \neq c$), and if $x_{c',j} < x_{c',j}$, then there exists column $i \in [m]$ such that $x_{i,j} > x_{i,j}$, thus $r'$ is also not maximal ($r' < j < r$ and $i \neq c$ since $q \le r' <j$).
Furthermore, since we took minimal $r$, the fact that $r' < r$ and $x_{c',r'} > y_{c',r'}$ implies that either $\hat{x}_{c',r'-1} \ge \hat{y}_{c',r'-1}$ or $r' = 1$.
In both cases, we get that $\hat{x}_{c',r'} > \hat{y}_{c',r'}$.
Thus, the fact that $x_{c',j} = y_{c',j}$, for each $j \in \{r'+1,r'+2,\dots,r-1\}$, implies that
\begin{equation}\label{eq:emd:intrinsicness:hat:1}
    \hat{x}_{c',j} > \hat{y}_{c',j}, \quad \mbox{for every } j \in \{r',r'+1,\dots,r-1\}.
\end{equation}

Similarly, since $q \le r'$ we get that $x_{c,j}=y_{c,j}$, for each $j \in \{r'+1,r'+2,\dots,r-1\}$.
Combining it with the fact that $\hat{x}_{c,r-1} < \hat{y}_{c,r-1}$ we get
\begin{equation}\label{eq:emd:intrinsicness:hat:2}
    \hat{x}_{c,j} < \hat{y}_{c,j}, \quad \mbox{for every } j \in \{r',r'+1,\dots,r-1\}.
\end{equation}

Building upon this, let us now construct matrix $X'$ with columns $x'_1,\dots,x'_m$ defined as follows:
\[
    x'_{i,j} =
    \begin{cases}
        x_{i,j} - 1, & \mbox{if } (i,j) \in \{(c,r), (c',r')\},\\
        x_{i,j} + 1, & \mbox{if } (i,j) \in \{(c,r'), (c',r)\},\\
        x_{i,j}, & \mbox{otherwise}. 
    \end{cases}
\]
Since $x_{c,r}>y_{c,r} \ge 0$ and $x_{c',r'}>y_{c,r} \ge 0$ we get that $x'_{i,j} \ge 0$ for every $i,j \in [m]$.
Moreover, the sums of entries in each row and column of $X'$ are still equal to $n$.
Thus, $X'$ is a position matrix of some election $E'_s$.

Let us now estimate $\demdpos(E'_s,E_t)$.
Observe that we have
\[
    \hat{x}'_{i,j} \! = \! 
    \begin{cases}
        \hat{x}_{i,j} \! - \! 1, &\!\!\!\!
            \mbox{if } (i,j) \!\in\!\{\! (c'\! ,\! r'),\! (c'\! ,\! r' \! + \! 1),\dots,\! (c,\! r \! -\! 1)\!\},\\
        \hat{x}_{i,j} \! + \! 1, &\!\!\!\!
            \mbox{if } (i,j) \!\in\!\{\! (c ,\! r'),\! (c,\! r' \! +\! 1),\dots,\! (c,\! r \! - \! 1)\!\},\\
        \hat{x}_{i,j}, & \mbox{otherwise}.
    \end{cases}
\]
Thus, we get that
$\emd(x'_{c'},y'_{c'}) = \emd(x_{c'},y_{c'}) - (r -r')$ from Eq.~\eqref{eq:emd:intrinsicness:hat:1}
and
$\emd(x'_c,y'_c) = \emd(x_c,y_c) - (r -r')$ from Eq.~\eqref{eq:emd:intrinsicness:hat:2}.
Since EMD of other columns do not change, we get
\[
    \demdpos(E'_s,E_t) \le \demdpos(E_s,E_t) - 2(r-r').
\]
Therefore, by the induction assumption, we get that there exists a sequence of elections
$E'_s = E_0,E_1,\dots,E_{k-1},E_k = E_t$
such that
\begin{multline}\label{eq:emd:intrinsicness:induction}
   \sum_{i=1}^k \demdpos(E_{i-1},E_i) \leq 2 \cdot \demdpos(E'_s,E_t) \le \\
   \le 2 \cdot \demdpos(E_s,E_t) - 4(r-r')
\end{multline}
and for each
$i \in [k]$, we have $\demdpos(E_{i-1},E_i) = 2$.

Hence, it remains to show that there exist elections
$E_s = E_{-k'},E_{-k'+1},\dots,E_{-1},E_{0}=E'_s$ such that
$k' \le 2(r-r')$ and
for every
$i \in \{-k',\dots,-2,-1\}$, we have $\demdpos(E_{i},E_{i+1}) = 2$.
With this and Eq.~\eqref{eq:emd:intrinsicness:induction}, the induction hypothesis will follow.
To this end, we prove the following claim:

\begin{claim}
\label{claim:emd:intrinsicness:1}
For every pair of elections $E_s$ and $E'_s$ such that $\dellpos(E_s,E'_s) = 4$ it holds that there exist elections
$E_s = E_{-k},E_{-k+1},\dots,E_1,E_0 = E'_s$ such that $k \le \demdpos(E_s,E'_s)$ and  for each
$i \in \{-k',\dots,-2,-1\}$, we have $\demdpos(E_{i},E_{i+1}) = 2$.
\end{claim}
\begin{proof}
We will prove this claim by induction on $\demdpos(E_s,E'_s)$.
If $\demdpos(E_s,E'_s)=2$, the statement follows trivially.
Hence, let us assume that $\demdpos(E_s,E'_s)>2$.

Take two arbitrary elections $E_s$ and $E'_s$ and let $X$ and $X'$ be their position matrices respectively.
Let $x_1,\dots,x_m$ be columns of $X$ and $x'_1,\dots,x'_m$ columns of $X'$.
Without loss of generality, let us assume that $\demdpos(E_s,E'_s) = \sum_{i \in [m]} \emd(x_i,x'_i)$;
otherwise we could rearrange the column order.
First, observe that $\dellpos(E_s,E'_s)=4$ if and only if there exist columns $c,c' \in [m]$ and rows $r,r' \in [m]$ such that
\[
    x'_{i,j} = \begin{cases}
        x_{i,j} - 1, & \mbox{if } (i,j) \in \{(c,r),(c',r')\},\\
        x_{i,j} + 1, & \mbox{if } (i,j) \in \{(c',r),(c,r')\},\\
        x_{i,j}, & \mbox{otherwise.}
    \end{cases}
\]
Without loss of generality assume that $r > r'$.
Then, observe that $\demdpos(E_s,E'_s) = 2(r - r')$.
Thus, since we assumed $\demdpos(E_s,E'_s) > 2$, we get that $r' < r-1$.
In what follows we construct elections $F$ and $F'$ such that $\demdpos(E_s,F) \le 2$, $\demdpos(F',E'_s) \le 2$, $\demdpos(F,F') = 2(r-r'-1)$, and $\dellpos(F,F') = 4$.
Then, the statement will follow from the inductive assumption.

Observe that there exists a column $c''$ such that $x_{c'',r-1} \ge 1$. 
If $c'' \neq c$,
let $Z$ be a matrix with columns $z_1,\dots,z_m$ defined as follows:
\[
    z_{i,j} = \begin{cases}
        x_{i,j} - 1, & \mbox{if } (i,j) \in \{(c,r),(c'',r-1)\},\\
        x_{i,j} + 1, & \mbox{if } (i,j) \in \{(c,r-1),(c'',r)\},\\
        x_{i,j}, & \mbox{otherwise.}
    \end{cases}
\]
If $c'' = c$, let us simply denote $Z = X$.
If $c'' \neq c$, then $z_{c,r} = x_{c,r} -1 = x'_{c,r} \ge 0$.
Also, $z_{c'',r-1} = x_{c'',r-1} - 1 \ge 0$.
Hence, $z_{i,j} \ge 0$, for every $i,j \in [m]$.
Moreover, sums of entries in every row and column of $Z$ are still equal $n$.
Thus, $Z$ is a position matrix of some elections, let us denote them by $F$.
If $c'' = c$, then let us simply denote $F = E_s$.

If $c'' \neq c$, then $\emd(z_c,x_c) = \emd(z_{c''},x_{c''})=1$.
Other columns are the same in both $Z$ and $X$, hence
\begin{equation*}
    \label{eq:lemma:emd-intrinsic:claim:1}
    \demdpos(E_s,F) \le 2.
\end{equation*}

Now, let us follow a similar construction for election $F'$.
If $c'' \neq c'$,
let $Z'$ be a matrix with columns $z'_1,\dots,z'_m$ defined as follows:
\[
    z'_{i,j} = \begin{cases}
        x'_{i,j} - 1, & \mbox{if } (i,j) \in \{(c',r),(c'',r-1)\},\\
        x'_{i,j} + 1, & \mbox{if } (i,j) \in \{(c',r-1),(c'',r)\},\\
        x'_{i,j}, & \mbox{otherwise.}
    \end{cases}
\]
If $c'' = c'$, let us simply denote $Z' = X'$.
If $c'' \neq c'$, then $z'_{c',r} = x_{c',r} -1 = x_{c,r} \ge 0$.
Also, $z'_{c'',r-1} = x'_{c'',r-1} - 1 = x_{c'',r-1} - 1 \ge 0$.
Hence, $z'_{i,j} \ge 0$, for every $i,j \in [m]$.
Moreover, sums of entries in every row and column of $Z'$ are still equal $n$.
Thus, $Z'$ is a position matrix of some elections, let us denote them by $F'$.
If $c'' = c'$, then let us simply denote $F' = E'_s$.

If $c'' \neq c'$, then $\emd(z'_{c'},x'_{c'}) = \emd(z'_{c''},x'_{c''})=1$.
Other columns are the same in both $Z'$ and $X'$, hence
\begin{equation*}
    \label{eq:lemma:emd-intrinsic:claim:2}
    \demdpos(E'_s,F') \le 2.
\end{equation*}

Now, observe that matrices $Z$ and $Z'$ can possibly differ only on entries in columns $c,c'$, and $c''$ and rows $r,r-1,$ and $r'$.
Specifically,
\begin{align*}
    z'_{c,r} &=    \!
    \begin{cases}
        x'_{c,r} = x_{c,r} - 1 = z_{c,r}, & \mbox{if } c \neq c'', \\
        x'_{c,r} + 1 = x_{c,r} = z_{c,r}, & \mbox{if } c = c''.
    \end{cases}\\
    z'_{c,r-1} &=    \!
    \begin{cases}
        x'_{c,r-1} \! = \! x_{c,r-1} \! = \! z_{c,r-1} \! - \! 1,  & \!\!\mbox{if } c \neq c'', \\
        x'_{c,r-1} \! - \! 1 \! = \! x_{c,r-1} \! - \! 1 \! = \! z_{c,r-1} \! - \! 1, &  \!\! \mbox{if } c = c''.
    \end{cases}\\
    z'_{c,r'} &= x'_{c,r'} = x_{c,r'} + 1 = z_{c,r'} + 1. \\
    z'_{c',r} &=  \!
    \begin{cases}
        x'_{c',r} -1 = x_{c',r} = z_{c',r}, & \mbox{if } c' \neq c'', \\
        x'_{c',r} = x_{c',r} + 1 = z_{c',r}, & \mbox{if } c' = c''.
    \end{cases}\\
    z'_{c',r-1} &= \! 
    \begin{cases}
        x'_{c',r-1} \! + \! 1 \! = \! x_{c',r-1} \! + \! 1 \! = \! z'_{c',r-1} \! + \! 1,  & \!\!\!\mbox{if } c' \!\neq c'', \\
        x'_{c',r-1} \! = \! x_{c',r-1} \! = \! z_{c',r-1} \! + \! 1, &  \!\!\!\mbox{if } c' \! = c''.
    \end{cases}\\
    z'_{c',r'} &= x'_{c',r'} = x_{c',r'} - 1 = z_{c',r'} -1.
\end{align*}
For column $c''$ it suffices to consider the case where $c'' \neq c$ and $c'' \neq c'$
since other cases has been considered above.
In such case,
\begin{align*}
    z'_{c'',r} &= x'_{c'',r} + 1 = x_{c'',r} + 1 = z_{c'',r}.\\
    z'_{c'',r-1} &= x'_{c'',r-1} - 1 = x_{c'',r-1} -1 = z_{c'',r-1}.\\
    z'_{c'',r'} &= x'_{c'',r'} = x_{c'',r'} = z_{c'',r'}.
\end{align*}
All in all, we get that
\[
    z'_{i,j} = \begin{cases}
        z_{i,j} - 1, & \mbox{if } (i,j) \in \{(c,r-1),(c',r')\},\\
        z_{i,j} + 1, & \mbox{if } (i,j) \in \{(c',r-1),(c,r')\},\\
        z_{i,j}, & \mbox{otherwise.}
    \end{cases}
\]
Hence, we get that $\ell_1(z_c,z'_c) = \ell_1(z_{c'},z'_{c'}) = 2$.
Since other columns of $Z$ and $Z'$ are identical, we get
$\dellpos(F,F') \le 4$.
If $\dellpos(F,F') = 0$, then also $\demdpos(F,F') =0$.
Then, since $\demdpos(E_s,E'_s) >2$, we get that $\demdpos(E_s,E'_s)=4$ and $\demdpos(E_s,F) = \demdpos(F,E'_s) =2$.
Hence, we can denote $E_{-2} = E_s$, $E_{-1} = F$, and $E_0 = E'_s$ and the induction hypothesis follows.

Thus, assume that $\dellpos(F,F') = 4$.
Then, observe that  $\emd(z_c,z'_c) = \emd(z_{c'},z'_{c'}) = r-r'-1$.
Since other columns of $Z$ and $Z'$ are identical, we get
$\demdpos(F,F') \le 2(r-r'-1)$.
Recall that $\demdpos(E_s,E'_s) = 2(r - r')$.
Therefore, from the inductive assumption there exist elections
$F = E_{-k},E_{-k+1},\dots,E_1,E_0 = F'$ such that $k \le \demdpos(E_s,E'_s) -2$ and for each $i \in \{-k,\dots,2,1\}$, we have $\demdpos(E_i,E_{i+1}) = 2$.
Since $\demdpos(E_s,F) \le 2$ and $\demdpos(E'_s,F') \le 2$, the induction hypothesis follows.
\end{proof}

Finally, observe that
$\dellpos(x_c,x'_c) = \dellpos(x_{c'},x'_{c'}) = 2$.
Other columns of $X$ and $X'$ are identical, thus
$\dellpos(E_s,E'_s) \le 4$.
If $\dellpos(E_s,E'_s) = 0$, then also $\demdpos(E_s,E'_s)=0$, from which the thesis follows by Eq.~\eqref{eq:emd:intrinsicness:induction}.
Hence, let us assume that $\dellpos(E_s,E'_s) = 4$.

Then, observe that $\emd(x_c,x'_c) = \emd(x_{c'},x'_{c'}) = r-r'$.
Other columns of $X$ and $X'$ are identical, thus
$\demdpos(E_s,E'_s) \le 2(r-r')$.
Hence, from Claim~\ref{claim:emd:intrinsicness:1}, we get that there exist elections
$E_s = E_{-k'},E_{-k'+1},\dots,E_{-1},E_{0}=E'_s$ such that
$k' \le 2(r-r')$ and
for every
$i \in \{-k',\dots,-2,-1\}$, we have $\demdpos(E_{i},E_{i+1}) = 2$.
Combined with Eq.~\eqref{eq:emd:intrinsicness:induction} this concludes the proof.
\end{proof}

\begin{lemma}\label{le:notint-EMD}
EMD-positionwise metric is not $\alpha$-intrinsic for any $\alpha < 2$.
\end{lemma}
\begin{proof}
Let us consider an identity election, $\ID$ over candidates $c_1,\dots, c_m$, with $n > m > 2$ votes: $c_1 \succ c_2 \succ \dots \succ c_m$.
Also, let us consider election $E$ that is obtained from $\ID$ by exchanging one of the votes by vote $c_m \succ c_2 \succ c_3 \succ \dots \succ c_{m-1} \succ c_1$.
Observe that position matrix of $E$ is:
\[
    \begin{bmatrix}
      n\! -\! 1 &  0  &  0  & \cdots &  0  &  1  \\
       0  & n   &  0  & \cdots &  0  &  0  \\
       0  &  0  &  n  & \cdots &  0  &  0  \\
      \vdots&\vdots&\vdots&\ddots&\vdots&\vdots \\
       0  &  0  &  0  & \cdots &  n  &  0  \\
       1  &  0  &  0  & \cdots &  0  & n\! -\! 1 \\
    \end{bmatrix}.
\]
Thus, it is optimal to match $c_i$ in $\ID$ to $c_i$ in $E$ for every $i \in [m]$ and we get that
$\demdpos(E,\ID)=2m-2$.

Let $k$ be the minimum number such that there are elections $E = E_0 ,E_1,\dots,E_{k-1}, E_k= \ID$ such that
for each
$i \in [k]$, $\dellpos(E_{i-1},E_i) = 2$.
Let us denote the position matrix of $E$ by $X^0$ (with columns $x^0_1,\dots,x^0_m$) and the position matrix of $\ID$ by $X^k$ (with columns $x^k_1,\dots,x^k_m$).
By Lemma~\ref{lemma:path-matching}, the existence of elections $E_0,E_1,\dots,E_{k-1},E_k$ is equivalent to the existence of position matrices $X^0,X^1,\dots,X^{k-1},X^k$ such that for each $s \in [k]$ the columns of matrix $X^s$ are $x^s_1,\dots,x^s_m$ and
\[
    \sum_{i \in [m]} \emd(x^{s-1}_i,x^{s}_i) = 2.
\]
In what follows, we will prove that $k \ge 2m - 3$.

To this end, for every matrix $X$, let us denote the position of the last nonzero entry in the first column of $X$ by $F(X)$ and by $G(X)$ the position of the first nonzero entry in the last column.
In particular, we have that $F(X^0)=m$ and $F(X^k)=1$ while $G(X^0)=1$ and $G(X^k)=m$.
We continue by proving the following claim:
\begin{claim}\label{claim:emd:intrinsicness:2:1}
For every $s \in [k]$, it holds that $|F(X^{s-1}) - F(X^s)| \le 1$ and $|G(X^{s-1}) - G(X^s)| \le 1$.
\end{claim}
\begin{proof}
Assume otherwise, i.e., there exists $s \in [k]$ such that $|F(X^{s-1}) - F(X^s)| > 1$ (or $|G(X^{s-1}) - G(X^s)| > 1$).
Observe that in such a case the EMD distance between the first (or last) columns of $X^{s-1}$ and $X^s$ is at least two, i.e., $\emd(x^{s-1}_1, x^{s}_1) \ge 2$ (or $\emd(x^{s-1}_m, x^{s}_m) \ge 2$).
On the other hand, two position matrices cannot differ on only one column.
In other words, there exists column $i \neq 1$ (or $i \neq m$) such that $\emd(x^{s-1}_i, x^{s}_i) > 0$.
Thus, $\sum_{i \in [m]} \emd(x^{s-1}_i,x^{s}_i) > 2$, which is a contradiction.
\end{proof}

Note that by Claim~\ref{claim:emd:intrinsicness:2:1} for each $j\in \{2,3,\dots,m\}$ there needs to be a matrix $X^s$ for some $s\in \{0,1,\dots,k-1\}$ with $F(X^s)=j$.
Thus, there exists a function, $f$, that
for every $j \in \{2,3,\dots,m\}$ returns the maximal index $s$ of a matrix, $X^s$, in which the position of the last nonzero entry in the first column is $j$.
Formally, let $f : \{2,3,\dots,m\} \rightarrow \{0,1,\dots,k-1\}$ such that $F(X^{f(j)})=j$ and for every $s \in \{0,1,\dots,k-1\}$ s.t. $F(X^s)=j$ it holds that $f(j) \ge s$.
Since $F(X^k) = 1$, applying \Cref{claim:emd:intrinsicness:2:1}, we get that:
\begin{equation}
\label{eq:emd:intrinsicness:2:1}
    F(X^s) < j, \quad \mbox{for every } s > f(j). \tag{\textasteriskcentered}
\end{equation}
Similarly, Claim~\ref{claim:emd:intrinsicness:2:1} implies that for each $j\in [m-1]$ there is a matrix $X^s$ for some $s\in \{0,1,\dots,k-1\}$ with $G(X^s)=j$.
Hence, there exists a function, $g$, that
for every $j \in [m-1]$ returns the maximal index $s$ of a matrix, $X^s$, in which the position of the first nonzero entry in the last column is $j$.
Formally, let $g : [m-1] \rightarrow \{0,1,\dots,k-1\}$ such that $G(X^{g(j)})=j$ and for every $s \in \{0,1,\dots,k-1\}$ s.t. $G(X^s)=j$ it holds that $g(j) \ge s$.
Since $G(X^k) = m$, applying \Cref{claim:emd:intrinsicness:2:1}, we get that:
\begin{equation}
\label{eq:emd:intrinsicness:2:2}
    G(X^s) > j, \quad \mbox{for every } s > g(j). \tag{\textasteriskcentered \textasteriskcentered}
\end{equation}

Next, by $I_f$ and $I_g$ let us denote the images of the functions $f$ and $g$, respectively.
As $f$ and $g$ are injective, $|I_f|=|I_g|=m-1$.
On the other hand, as we show in the next claim the intersection of $I_f$ and $I_g$ is limited.
\begin{claim}\label{claim:emd:intrinsicness:2:2}
If $s \in I_f \cap I_g$, then $F(X^s) - G(X^s) = 1$.
\end{claim}
\begin{proof}
Let us take an arbitrary $s \in I_f \cap I_g$ and denote $j = F(X^s)$ and $j' = G(X^s)$.
Since $s \in I_f$, this means that $f(j)=s$.
Thus, from~\eqref{eq:emd:intrinsicness:2:1} we get that $F(X^{s+1})<j$, which implies that $x^{s+1}_{1,j} = 0$
(the last nonzero entry in the first column must be at position smaller than $j$).
On the other hand, since $F(X^s)=j$, we have $x^{s}_{1,j} > 0$.
Hence, $\emd(x^s_1,x^{s+1}_1) \ge 1$.
Similarly, the fact that $s \in I_g$ implies $g(j')=s$.
Thus, from~\eqref{eq:emd:intrinsicness:2:2} we get $G(X^{s+1})>j'$, which implies that $x^{s+1}_{m,j'} = 0$
(the first nonzero entry in the first column must be at position greater than $j'$).
However, since $G(X^s)=j'$, we have $x^{s}_{m,j'} > 0$.
Hence, $\emd(x^s_m,x^{s+1}_m) \ge 1$.

Now, observe that the fact that $\sum_{i \in [m]} \emd(x^{s}_i,x^{s+1}_i) = 2$ implies that there exist columns $c,c'$ and row $r$ such that
\[
    x^{s+1}_{i,j} =
    \begin{cases}
        x^s_{i,j} - 1, & \mbox{if } (i,j) \in \{(c,r),(c',r-1)\}, \\
        x^s_{i,j} + 1, & \mbox{if } (i,j) \in \{(c',r),(c,r-1)\}, \\
        x^s_{i,j}, & \mbox{otherwise.}
    \end{cases}
\]
Hence, it must hold that $c=1$, $c'=m$, $r = j$, and $r-1 = j'$.
Thus,
\[
    F(X^s) - G(X^s) = j - j' = r - (r-1) = 1.
\]
\end{proof}

From~\eqref{eq:emd:intrinsicness:2:1} we get that the sequence $f(2),f(3),\dots,f(m)$ is strictly decreasing.
On the other hand, \eqref{eq:emd:intrinsicness:2:2} implies that $g(1),g(2),\dots,g(m-1)$ is strictly increasing.
Therefore, there can exists at most one $j$ and one $j'$ such that
$f(j) = g(j'-1)$.
Combining this with Claim~\ref{claim:emd:intrinsicness:2:2}, we get that
$|I_f \cap I_g| \le 1$.
Thus, by the inclusion–exclusion principle, we get
\[
    | I_f \cup I_g| \ge 2m - 3.
\]
Clearly, $k \ge | I_f \cup I_g|$, hence
\[
    k \ge 2m-3.
\]

Finally, from the definition of the level of intrinsicness, we get that
\(
    \alpha \ge k \cdot d_{min} / \demdpos(E,\ID).
\)
Thus,
\[
    \alpha \ge 2 \cdot \frac{2m - 3}{2m - 2} = 2 - \frac{1}{m - 1}.
\]
Since $m$ can be arbitrarily large, we get that $\alpha \ge 2$.
\end{proof}

\medskip

\Cref{thm:intrinsic} directly follows from \Cref{le:intrin,prop:l1-pos:intrinsic-degree,prop:emd-pos:intrinsic-degree,le:not-intrinsic,le:notint-l1,le:notint-EMD}.

\clearpage

\end{document}